\documentclass[prodmode,acmec,table]{acmsmall-report}
\acmVolume{X}
\acmNumber{X}
\acmArticle{1}
\acmYear{2017}
\acmMonth{08}
\usepackage[numbers]{natbib}
\usepackage[nodayofweek]{datetime}

\newdateformat{mydate}{\twodigit{\THEDAY}{ }\shortmonthname[\THEMONTH], \THEYEAR}

\usepackage{amssymb}

\RequirePackage{ifthen,calc}
\newsavebox{\ffbox}\newlength{\ffboxlen}
\newcommand{\todo}[1]{{\sbox{\ffbox}{\textbf{TODO:}\ \textit{{#1}}\ \textbf{:ODOT}}
  \settowidth{\ffboxlen}{\usebox{\ffbox}}
  \addtolength{\ffboxlen}{-5mm}
  \ifthenelse{\ffboxlen>\linewidth}{%
  \noindent\marginpar{$>>>>$}\textbf{TODO:}\ \textit{{#1}}\ \textbf{:ODOT}\marginpar{$<<<<$}}{%
  \noindent\marginpar{$>><<$}\textbf{TODO:}\ \textit{{#1}}\ \textbf{:ODOT}}}}

\usepackage[usenames,dvipsnames]{xcolor}
\usepackage{array,xspace,multirow,hhline,tabularx,booktabs,fixltx2e}
\usepackage[margin=1in]{geometry}
\usepackage{amsmath,amssymb,amsfonts}

\usepackage{subfig}
\usepackage{lscape}

\DeclareSymbolFont{largesymbolsA}{U}{txexa}{m}{n}
\DeclareMathSymbol{\varprod}{\mathop}{largesymbolsA}{"10}

\usepackage{tikz}
\usetikzlibrary{positioning,chains,fit,shapes,calc}
\usetikzlibrary{arrows}
\usetikzlibrary{decorations.pathreplacing}
\usetikzlibrary{patterns}
\usetikzlibrary{calc}
\tikzset{monRect/.style={draw,rectangle,minimum width=10mm,minimum height=5mm}}

\usepackage{paralist}
\usepackage{algorithm}
\usepackage{algorithmicx}
 \usepackage{algpseudocode}
 \usepackage{algcompatible}
  \usepackage{nicefrac}
\algnewcommand\algorithmicreturn{\textbf{return }}
\algnewcommand\RETURN{\State \algorithmicreturn}%

\usepackage{algpseudocode}

\newcommand\Algphase[1]{%
\vspace*{-.7\baselineskip}\Statex\hspace*{\dimexpr-\algorithmicindent-2pt\relax}\rule{\textwidth}{0.4pt}%
\Statex\hspace*{-\algorithmicindent}\textbf{#1}%
\vspace*{-.7\baselineskip}\Statex\hspace*{\dimexpr-\algorithmicindent-2pt\relax}\rule{\textwidth}{0.4pt}%
}

\newlength{\wordlength}

\usepackage{enumitem}
\setenumerate[1]{label=\rm(\it{\roman{*}}\rm),ref=({\it\roman{*}}),leftmargin=*}

\newtheorem{claim}{Claim}

\usepackage[mathcal]{euscript}

\newcommand{\midd}{\mathrel{:}}

\author{HARIS AZIZ\affil{Data61, CSIRO and UNSW Australia} \and SIMON MACKENZIE\affil{Carnegie Mellon University, USA}}

\title{A Discrete and Bounded Envy-free Cake Cutting Protocol\\ for Any Number of Agents}   

\begin{abstract}
We consider the well-studied cake cutting problem in which the goal is to find an envy-free allocation based on queries from $n$ agents.
The problem has received attention in computer science, mathematics, and economics. 
It has been a major open problem whether there exists a discrete and bounded envy-free protocol. 
We resolve the problem by proposing a discrete and bounded envy-free protocol for any number of agents.
The maximum number of queries required by the protocol is $n^{n^{n^{n^{n^n}}}}$.
We additionally show that even if we do not run our protocol to completion, it can find in at most $n^3{(n^2)}^n$ queries a partial allocation of the cake that achieves  proportionality (each agent gets at least $\nicefrac{1}{n}$ of the value of the whole cake) and envy-freeness. Finally we show that an envy-free partial allocation can be computed in at most $n^3{(n^2)}^n$ queries such that each agent gets a connected piece that gives the agent at least $\nicefrac{1}{3n}$ of the value of the whole cake. 
\end{abstract}



\begin{document}
\begin{bottomstuff}
Emails: \texttt{haris.aziz@data61.csiro.au, simonm@andrew.cmu.edu} 
\end{bottomstuff}
\sloppy

\maketitle

\section{Introduction}

\begin{quote}
  ``Despite intense efforts over decades, up to this date no one has succeeded in finding a finite bounded cake cutting protocol that guarantees envy-freeness for any number of players.''--- Lindner and Rothe~\cite{LiRo15a}.
\end{quote}


The existence of a discrete and bounded envy-free cake cutting protocol has remained a major open problem for at least two decades. In this paper, we settle the problem by presenting a general discrete envy-free protocol for any number of agents.
The protocol is for the cake cutting setting that is a versatile mathematical model  for allocation of a heterogeneous divisible good among multiple agents with possibly different preferences over different parts of the cake. The main applications of cake cutting are fair scheduling, resource allocation, and conflict resolution~\cite{DQS12a}. Cake cutting has been extensively studied within computer science~\cite{Proc15a} and the social sciences~\cite{Thom07a}. Since various important divisible resources such as time and land can be captured by cake cutting, the problem of fairly dividing the cake is a fundamental one within the area of fair division and multiagent resource allocation~\cite{BrTa96a,Guo15a,Proc12a,RoWe98a,Stew99a,Su99a,Thom07a}.

A cake is represented by an interval $[0,1]$ and each of the $n$ agents has a valuation function over pieces of the cake that specifies how much that agent values a particular subinterval. The main goal is to divide the cake fairly. 
Among various fairness concepts proposed by social scientists, a prominent one is \emph{envy-freeness}.
An allocation is envy-free if no agent would prefer to take another agent's allocation instead of his own. 
Although an envy-free allocation is guaranteed to exist even with $n-1$ cuts~\cite{Su99a}\footnote{Su~\cite{Su99a} pointed out that the existence of an envy-free cake allocation can be shown via an interesting connection with Sperner's Lemma.}, \emph{finding} an envy-free allocation is a challenging problem which has been termed ``one of the most important open problems in 20th century mathematics'' by Garfunkel~\cite{Garf88a}.

\paragraph{Motivation and Contribution}
Since the valuations of agents over the cake can be complex, eliciting each agent's complete valuations function over the cake is computationally infeasible.
A natural approach in cake cutting protocols is to query agents about their valuations of different portions of the cake and based on these queries propose an allocation.
A cake cutting protocol is \emph{envy-free} if each agent is guaranteed to be non-envious if he reports his real valuations. 
If a protocol is envy-free, then an honest agent will not be envious even if other agents misreport their valuations. 
For the case of two agents, the problem has a well known solution in the form of the \emph{Divide and Choose} protocol: one agent is asked to cut the cake into equally preferred pieces and the other agent is asked to choose the preferred piece. 
For the case of three agents, an elegant and bounded protocol was independently discovered by John L. Selfridge and John H. Conway around 1960~\cite{BrTa96a}. Since then, an efficient general envy-free protocol for any number of agents has eluded mathematicians, economists, and computer scientists.


In 1995, Brams and Taylor~\cite{BrTa95a} made a breakthrough by presenting an envy-free protocol for \emph{any} number of agents~\cite{Hive95a}. Although the protocol is guaranteed to terminate in finite time, there is one  drawback of the protocol: the running time or number of queries and even the number of cuts required is unbounded even for four agents. In other words, the number of queries required to identify an envy-free allocation can be arbitrarily large for certain valuations functions. 
Procaccia~\cite{Proc15a} terms unboundedness as a ``serious flaw''.
Brams and Taylor were cognizant of their protocol's drawback and explicitly mentioned the problem of proposing a bounded envy-free protocol even for $n=4$. 
Lindner and Rothe~\cite{LiRo09a} write that ``the development of finite bounded envy-free cake cutting protocols still appears to be out of reach, and a big challenge for future research.'' 
The problem has remained open and has been highlighted in several works~\cite{BaTa95a,BrTa95a,BrTa96a,BKM05a,EdPr06a,HHA15a,KLP13a,Proc12a,Proc15a,LiRo15a,RoWe98a,SaWa09a}. 
Saberi and Wang~\cite{SaWa09a} term the problem as ``one of the most important open problems in the field''. 
Procaccia~\cite{Proc12a} goes one step further and calls it an interesting challenge within theoretical computer science: ``Since the 1940s, the computation of envy-free cake divisions has baffled many great minds across multiple disciplines. Settling this problem once and for all is an important challenge for theoretical computer science.'' 
In recent work, Aziz and Mackenzie~\cite{AzMa16a} proposed a bounded and discrete envy-free protocol for $n=4$. Despite this progress, the general problem for any number of agents had remained unaddressed. 

In this paper, \emph{we present a discrete envy-free protocol for any number of agents that requires a bounded number of queries and hence a bounded numbers of cuts of the cake} thereby closing the central open problem in the field of cake cutting. 
The bound for the number of queries is $n^{n^{n^{n^{n^n}}}}$.
Previously, no discrete protocol was known that even uses a bounded number of cuts. 
Apart from using some classic ideas in cake-cutting,   
our protocol requires some radically different techniques. 
Since only three general finite and discrete cake cutting protocols have been presented in the literature~\cite{BrTa95a,RoWe97a,Pikh00a}, our protocol is another addition to the list of finite envy-free protocols and is of interest even if boundedness is not a critical property. More generally, the protocol provides a new constructive argument for existence of envy-free allocations. Our result is surprising because leading experts in the area have conjectured that a bounded and discrete envy-free protocol does not exist: ``It is natural to ask whether the envy-free cake cutting problem is inherently difficult: Is it provably impossible to design a bounded envy-free cake cutting algorithm?''\cite{Proc15a}. Previously, mathematician Ian Stewart had commented that ``However, no discrete procedure with a bounded number cuts (however large) is known for four players, and such schemes probably don't exist''\cite{Stew06a}.

We additionally show that even if we do not run our protocol to completion, it can find in at most $n^3{(n^2)}^n$ queries a partial allocation of the cake that achieves proportionality (with respect to the whole cake) and envy-freeness. Finally, we show that even if we do not run our protocol to completion, it can be tweaked to find in at most $n^3{(n^2)}^n$ queries an envy-free partial allocation of the cake in which each agent gets a \emph{connected piece} that gives the agent $\nicefrac{1}{3n}$ of the value of the whole cake. For the additional two results, we only need one of our subprotocols called the SubCore Protocol that we  define later.
%
%

Apart from introducing new techniques, our protocol relies on some ideas from previous protocols. 
At the heart of our protocol is the idea of domination or `irrevocable advantage'~\cite{BrTa95a}. An agent $i$ dominates another agent $j$ if he is not envious of $j$ even if the unallocated cake is given to $j$. Note that if agents in subset $S\subset N$ dominate all the agents in $N\setminus S$, then agents in $S$ can simply let the agents in $N\setminus S$ worry about the remaining cake since they will not be envious if a single agent in $N\setminus S$ got all the remaining unallocated cake. 
Our protocol uses some high level ideas from our previous protocol for 4 agents: 
\begin{inparaenum}[(i)]
 \item We use a `Core Protocol' over the unallocated cake (aka residue) with one agent as the specified cutter that results in an envy-free partial allocation and a smaller residue. The Core Protocol is called repeatedly on the latest residue so as to make the residue even smaller. As the new calls to the Core Protocol  are made, the partial allocation remains envy-free. 
 \item Every time agents make trims or evaluate pieces of cake, we do not lose such information and in fact keep track of it. 
 \item We make agents exchange some parts of their already allocated pieces to obtain new dominations. The exchange was referred to as a \emph{permutation} in  \cite{AzMa16a}.
\end{inparaenum}
In some steps, our protocol also uses a simple fact: envy-freeness implies proportionality with respect to the allocated cake (i.e., each agent gets at least $\nicefrac{1}{n}$ value of the allocated cake).

Throughout the overall protocol, we have cake that has been allocated and some cake that is the unallocated residue. During the protocol, we ensure that the cake that has been allocated to the agents has been done so in an envy-free manner.    
Overall, our envy-free protocol (\emph{Main Protocol}) takes as input the cake and a set of agents and returns an envy-free allocation of the whole cake. In order to obtain the envy-free allocation, the Main Protocol  calls three other protocols. It repeatedly calls the \emph{Core Protocol} a bounded number of times. The Core Protocol is the work-horse of the overall protocol that is called repeatedly to further allocate unallocated cake in any envy-free manner. The other protocols---GoLeft and Discrepancy are used to make one set of agents dominate the others so that the problem reduces to finding an envy-free allocation for fewer agents and the Main Protocol can be called again for this sub-problem. 

After calling the Core Protocol many times, 
if there is still some unallocated cake and there is sufficient difference in the estimation of the value of a cake piece between some agents, then the Discrepancy Protocol is called by the Main Protocol. 
The goal of the Discrepancy Protocol is to either exploit the difference in valuations or else to ensure that on a very rough scale, agents have similar valuations. 
If after calling the Discrepancy Protocol, the cake still has not been allocated in an envy-free way, the GoLeft Protocol is called by the Main Protocol.
The GoLeft Protocol is used to reallocate pieces of the cake among the agents after additionally adding some crumbs from the residue. The reallocation is useful to obtain new dominations until one set of agents dominate the remaining agents which means that we have reduced the problem to allocating the remainder of the cake among the dominated agents. 
The different protocols that we present are summarized in Figure~\ref{fig:subprotocols}.

\begin{figure}[h!]
\centering
\scalebox{1}{
\begin{tikzpicture}[%
          auto,
          block/.style={
            rectangle,
            draw=blue,
            thick,
            text width=6em,
            align=center,
            rounded corners,
            minimum height=2em
          },
          block1/.style={
            rectangle,
            draw=blue,
            thick,
            text width=6em,
            align=center,
            rounded corners,
            minimum height=2em
          },
          line/.style={
            draw,thick,
            -latex',
            shorten >=2pt
          },
          cloud/.style={
            draw=red,
            thick,
            ellipse,
            fill=red!20,
            minimum height=1em
          }
        ]
\centering
\draw ( 0,-3  ) node[block] (C)    {Core Protocol (Algorithm~\ref{algo:C})};
\draw ( 0,-5.5) node[block] (SC)   {SubCore Protocol (Algorithm~\ref{algo:SC})};
\draw (8,-3  ) node[block] (GL)   {GoLeft Protocol (Algorithm~\ref{algo:GL})};
\draw ( 4,-3  ) node[block] (DISC) {Discrepancy Protocol (Algorithm~\ref{algo:D})};
\draw ( 0,-1  ) node[block] (M)    {Main Protocol (Algorithm~\ref{algo:M})};

\path[draw,thick,->] (C) -- (SC);
\path[draw,thick,->] (M) -- (DISC);
\path[draw,thick,->] (M) -- (GL);
\path[draw,thick,->, bend left] (GL) -- (M);
\path[draw,thick,->] (M) -- (C);
\path[draw,thick,->] (DISC) -- (C);
\path[] (SC) edge [loop right,thick] node {} (SC);
\path[] (M) edge [loop left,thick] node {} (M);
\end{tikzpicture}
}
\caption{The envy-free protocol for $n$ agents relies on various protocols. A protocol points to another protocol if it calls the other one. A protocol has a self-loop if it calls itself recursively. The Main Protocol is called to get an envy-free allocation.}
\label{fig:subprotocols}
\end{figure}
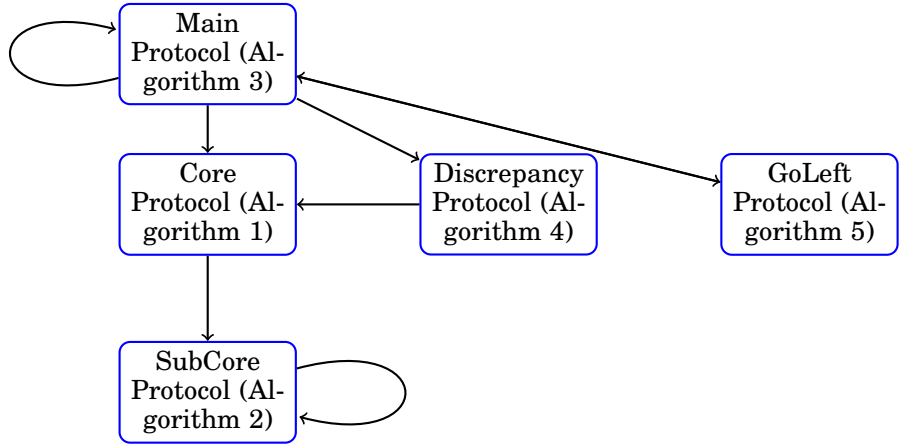

Next, we informally summarize what input each protocol takes as a blackbox and then what it outputs.
\begin{itemize}
	\item \textbf{Main Protocol}
	\begin{itemize}
		\item Input: Cake $R$ and a set of agents $N$.
 
        \item Output: An envy-free allocation that completely allocates $R$ among agents in $N$.
	\end{itemize}
	\item \textbf{Core Protocol}
	\begin{itemize}
	           \item Input: Specified cutter agent, agent set $N$ and unallocated cake $R$.
	           \item Output: An envy-free allocation of cake $R'\subset R$ for agents in $N$ and updated unallocated cake $R\setminus R'$.
	\end{itemize}
           \item \textbf{SubCore Protocol}
	\begin{itemize}
		  \item Input: Cake $R$ cut into $n$ pieces to be allocated among agents in set $N'\subset N$ with $n'=|N'|$. Additionally, each agent $i\in N'$ has a benchmark value $b_i$. 
                        \item Output: An envy-free partial allocation of $R$ for agents in $N'$ in which each agent $j\in N'$ gets a piece of value at least $b_j$.
	\end{itemize}
	\item \textbf{Discrepancy Protocol}
	\begin{itemize}
		\item Input: Residue $R$, a specific piece in $R$, a specific value given by agent $u$, a set of pieces of cake other than $R$, and a set of agents $N$.
\item Output: Possibly modified residue $R$, a Boolean value called DISCREPANCY, and a partition of $N$ into sets $D'$ and $D''$.

	\end{itemize}

	\item \textbf{GoLeft Protocol}
	\begin{itemize}
		  \item Input: A set $N$, a set of `snapshots' (disjoint envy-free allocation for the agents), a set of corresponding `extracted pieces' for each piece in the disjoint allocations, and residue $R$.
		\item Output:  A set of agents $A \subset N$ such that all agents in $N \setminus A$ dominate all agents in $A$.
	\end{itemize}
\end{itemize}

\paragraph{New Techniques and Ideas}    
    
Despite some high level similarities with the approach of the recent protocol for four agents~\cite{AzMa16a}, our general protocol requires a number of new concepts and techniques. These concepts and techniques will be formalized and explained in the later sections but we mention them briefly and informally. 

We define the notion of a \emph{significant advantage or bonus} that informally means that an agent $i$ feels his absolute advantage over another agent is of a high enough fraction of his value of the unallocated cake. 
We then make use of our Core Protocol that not only reduces the size of the unallocated residue but also guarantees that a given agent has a significant advantage over someone else. 
The Core Protocol can be used to further allocate the residue  so that an agent's significant advantage over another agent translates into dominance in a bounded number steps. 
This approach is one of the key ideas for obtaining a bounded protocol rather than simply a finite protocol.

We rely on a notion of discrepancy that is based on the idea of how differently agents view a piece of cake. 
It is well known in fair division that the more different the agents' valuations are, the easier it is to achieve fairness. 
We try to use the structure of agents that have similar valuations at a very coarse level. In case they have high enough discrepancy in valuations over two parts of the unallocated cake, then we simply break the problem into two subproblems where we divide each part in any envy-free manner among those agents who value it highly compared to the other part.


Another key idea to achieve a bounded protocol is that we run the Core Protocol enough times to ensure there are a sufficient number of Core Protocol allocation snapshots (allocation outcomes of various calls of the Core Protocol) that are \emph{isomorphic}.\footnote{The notion of isomorphic snapshots will be defined later.} We will keep track of these snapshots to later make more agents dominate each other. 

Further operations such as exchanges of allocated pieces are then implemented on such isomorphic Core snapshots by exploiting their structure. In order to achieve boundedness of the protocol, we also define and work with four carefully chosen bounds that are functions of the number of agents. The biggest of the four bounds is the maximal  number of query operations required to complete the protocol.

We use a new technique called \emph{extraction} whereby part of the residue is systematically trimmed off so that it can potentially be combined with some agent's allocation in the Core Protocol. When a particular extracted piece of cake is  given to an agent in addition to his allocated piece in a single snapshot of the Core Protocol, we refer to this as \emph{attachment}. Attachment is helpful to systematically enable more agents to exchange each others' pieces in snapshots of the Core Protocol without causing any agent to be envious. The exchange of pieces consequently allows for `permutations'  and helps achieve agents dominate each other.

In order to methodically track the possible permutations/exchanges, we keep track of a \emph{permutation graph} in which nodes correspond to agents and an agent $i$ points to agent $j$ if $i$ is willing to replace his allocation in a subset of isomorphic snapshots of the Core Protocol with $j$'s allocated piece plus a tiny bit more cake without causing envy. The tiny bit more cake to attract agent $i$ can in principle cause agents to be envious so we have to compensate the other agents elsewhere so that envy is not introduced. 
The attachments and exchanges take place in the GoLeft Protocol. The interplay between the permutation graph and the working set of isomorphic Core snapshots is technically one of the most interesting parts of the overall protocol and is the central part of the GoLeft Protocol.




\paragraph{Related Work}

Cake cutting problems originated in the 1940's when famous mathematicians such as Banach, Knaster, and Steinhaus initiated serious mathematical work on the topic of fair division.\footnote{Hugo Steinhaus presented the cake cutting problems to the mathematical and social science communities on Sep. 17, 1947, at a meeting of the Econometric Society in Washington, D.C.~\cite{RoWe98a,Stei48a}.}
Since then, the theory of cake cutting algorithms has become a full-fledged field with at least three books written on the topic~\cite{Barb05a,BrTa96a,RoWe98a}. 
The central problem within cake cutting is finding an envy-free allocation~\cite{GaSt58a,Stew99a}.

Since the earliest works, mathematicians have been interested in the complexity of cake cutting. Steinhaus~\cite{Stei48a} wrote that ``Interesting mathematical problems arise if we are to determine the minimal number of cuts necessary for fair division.''
When formulating efficient cake cutting protocols, a typical goal is to minimize the number of cuts while ignoring the number of valuations queried from the agents. In principle, the actual complexity of a problem or a protocol depends on the number of queries. 
When considering how efficient a protocol is, it is useful to have a formal query model for cake cutting protocols. Robertson and Webb~\cite{RoWe98a} formalized a simple query model in which there are two kinds of queries: Evaluation and Cut. In an Evaluation query, an agent is asked how much he values a subinterval. In a Cut query, an agent is asked to identify an interval, with a fixed left endpoint, of a particular value. 
Although, the query model of Robertson and Webb~\cite{RoWe98a} is very simple, it is general enough to capture all known protocols in the literature. Note that if the number of queries is bounded, it implies that the number of cuts is bounded in the Robertson and Webb model. The protocol that we present in this paper uses a bounded number of queries in the  Robertson and Webb model. 

There are various notions of fairness in the cake-cutting literature. Two of the most prominent ones are envy-freeness and proportionality. Proportionality requires that each agent gets at least $\nicefrac{1}{n}$ value of the whole cake. An envy-free allocation of the whole cake also satisfies proportionality.
Finding a proportional allocation is an easier task than finding an envy-free allocation. There exists a general discrete protocol for finding a proportional allocation with connected pieces that takes $O(n\log n)$ queries~\cite{EvPa84a}.

There is not too much known about the existence of a bounded envy-free protocol for general $n$ except that any envy-free cake cutting algorithm requires $\Omega(n^2)$ queries in the Robertson-Webb model~\cite{Proc09a,Proc15a}.
Also,  for $n\geq 3$, there exists no finite envy-free cake cutting algorithm that outputs \emph{contiguous} allocations in which each agent gets a connected piece with no gaps~\cite{Stro08a}. 
Brams et al. \cite{BTZ97a} and Barbanel and Brams~\cite{BaBr04a} presented envy-free protocols for four agents that require 13 and 5 cuts respectively. However, the protocols are not only unbounded but also not finite since they are \emph{continuous} protocols that require the notion of a \emph{moving knife}.
An alternative approach is to consider known bounded protocols and see how well they perform in terms of envy-freeness~\cite{LiRo09a}. Apart from the unbounded Brams and Taylor envy-free protocol for $n$ agents, there are other general envy-free protocols by Robertson and Webb~\cite{RoWe97a} and Pikhurko~\cite{Pikh00a} that are also unbounded. Gasarch~\cite{Gasa15a} compared the complexity of the three general unbounded envy-free protocols in the literature. 
Aziz and Mackenzie~\cite{AzMa16a} proposed a bounded and discrete envy-free protocol for $n=4$. 
We use similar ideas but generalizing to the case of any number of agents requires novel subroutines and significantly more abstract arguments. 

 
There are positive algorithmic results concerning envy-free cake cutting when agents have restricted valuations functions~\cite{Bran15b,CLPP11a,DQS12a} or when some part of the cake is left unallocated~\cite{HHA15a,SaWa09a}. 
There has also been work on \emph{strategyproof} cake cutting protocols for restricted valuation functions~\cite{AzYe14a,CLPP12a,MaNi12a} as well as strategic aspects of protocols~\cite{BrMi15a}. 

\section{Preliminaries}
\label{prel}
 \paragraph{Model}
	We consider a cake which is represented by the interval $[0,1]$.
A \emph{piece of cake} is a finite union of disjoint subintervals  of $[0,1]$. 
We will assume the standard assumptions in cake cutting.
Each agent in the set of agents $N=\{1,\ldots, n\}$ has his own valuation over subintervals of the interval $[0,1]$. 
We denote by $V_i(x,y)$ the value agent $i$ has for interval $[x,y]$.
The valuations are 
\begin{inparaenum}[(i)]
\item \emph{non-negative}: $V_i(Y)\geq 0$ for any subinterval $Y$; 
\item \emph{additive}: for all disjoint subintervals $Y,Y'$, $V_i(Y\cup Y')=V_i(Y)+V_i(Y')$;  and
\item \emph{divisible}, i.e., for every subinterval $Y$ and $0\leq \lambda \leq 1$, there exists $Y'\subseteq Y$ with $V_i(Y')=\lambda{V_i(Y)}$.
\end{inparaenum}
It follows from the divisibility property that the valuations are non-atomic i.e., $V_i(x,x)=0$.

	We will typically denote an allocation by $X=(X_1,\ldots, X_n)$ where $X_i$ is the cake allocated to agent $i$. Two main criteria of fairness are envy-freeness and proportionality. 
	
	\begin{definition}[Envy-free allocation]
	An allocation $X=(X_1,\ldots, X_n)$ is \emph{envy-free} if
$V_i(X_i)\geq V_i(X_j)$ for each $i,j\in N$.
	\end{definition}
	
		\begin{definition}[Proportional allocation]
		An allocation $X=(X_1,\ldots, X_n)$ is \emph{proportional} if
	$V_i(X_i)\geq \frac{1}{n}V_i([0,1])$ for each $i\in N$.
		\end{definition}
	
In order to ascertain the complexity of a protocol, Robertson and Webb presented a computational framework in which two kinds of queries can be made to the agents: (1) for given $x\in [0,1]$ and $r\in \mathbb{R}^+$, a {\sc Cut} query asks an agent to return a point $y\in [0,1]$ such that $V_i([x,y])=r$ (2) for given $x,y \in [0,1]$, {\sc Evaluate} query asks an agent to return a value $r\in \mathbb{R}^+$ such that $V_i([x,y])=r$. A cake cutting protocol specifies how agents interact with queries and cuts. All well known cake cutting protocols can be analyzed in terms of the number of queries required to return a fair allocation. A cake cutting protocol is \emph{bounded}  if the number of queries required to return a solution is bounded by a function of $n$ irrespective of the valuations of the agents.


\paragraph{Terms and Conventions}
We now define some terms and conventions that we will use in the paper. An allocation is \emph{partial} if it does not necessarily allocate the whole cake. 
Given an envy-free partial allocation $X$ of the cake and an unallocated residue $R$, we say that agent $i$ \emph{dominates} agent $j$ if $i$ does not become envious of $j$ even if all of $R$ were to be allocated to $j$:
 \[V_i(X_i)\geq V_i(X_j)+V_i(R).\]


We say that an agent gets a piece \emph{partially}, if he does not necessarily get it completely. 
In the cake cutting protocols that we will describe, an agent may be asked to trim a piece of cake so that its value equals the value of a less valuable piece.  
Agents will be asked to trim various pieces of the cake to make a given piece equal to the value of another piece.  In Figure~\ref{fig:trim}, we outline the idea of trimming a piece to equal the value of another piece.
When an agent trims a piece of cake, he will trim it from the left side: the main piece (albeit trimmed) will be on the right side. The piece minus the  trim will be called the partial main piece. 
When we say that an agent is \emph{guaranteed} to get his second/third/etc. most favoured piece, this guarantee is based on the \emph{ordinal} preferences of the agents over the pieces. By ordinal, we mean that agents simply give a weak ordering over the pieces but do not tell the exact cardinal utility difference between two pieces. If an agent is indifferent between the top three pieces, then we will still say that the agent is guaranteed to get his third most valued piece.

\begin{figure}[h!]
\centering
\begin{tikzpicture}[yscale=0.4]
\draw (0mm, 0mm) rectangle (50mm, 10mm);
\node[blue] at (60mm, 10mm) {1};
\node[red] at (65mm, 10mm) {2};
\node[green] at (70mm, 10mm) {3};

\draw[green] (20mm, 13mm) -- (20mm, -3mm);
\draw[green] (35mm, 13mm) -- (35mm, -3mm);

	[fill,red] (0mm, 5mm) rectangle (19.8mm, 10mm);
	[fill,blue] (0mm, 0mm) rectangle (19.8mm, 5mm);
\draw[dashed](0mm,0mm)--(-10mm,-10mm);
\draw[dashed](20mm,0mm)--(10mm,-10mm);
\draw (-10mm,-20mm) rectangle (10mm,-10mm);

\draw[blue] (-5.5mm, -7mm) -- (-5.5mm, -23mm);
\draw[red] (-2.5mm, -7mm) -- (-2.5mm, -23mm);
\node[blue] at (-6.4mm, -7.5mm) {1};
\node[red] at (-1.2mm, -7.5mm) {2};
\node[] at (-4mm, -15mm) {$\gamma$};

\draw[dashed](-5mm,-20mm)--(-3mm,-30mm);
\draw[dashed](10mm,-20mm)--(13mm,-30mm);
\draw (-3mm,-40mm) rectangle (13mm,-30mm);

\draw[dashed](-5mm,-20mm)--(-8mm,-30mm);
\draw[dashed](-10mm,-20mm)--(-13mm,-30mm);
\draw (-13mm,-40mm) rectangle (-8mm,-30mm);

\node[brown] at (-11mm, -35mm) {};
\node[red] at (5mm, -35mm) {2};

\draw[dashed](20mm,0mm)--(30mm,-30mm);
\draw[dashed](50mm,0mm)--(60mm,-30mm);
\draw (30mm,-40mm) rectangle (60mm,-30mm);
\draw[green](45mm,-27mm)--(45mm,-43mm);

\node[blue] at (37mm, -35mm) {1};
\node[green] at (52mm, -35mm) {3};

\draw(-2mm,-43mm)--(59mm,-43mm);
\draw((-2mm,-43mm)--((-3mm,-41mm);
\draw((59mm,-43mm)--(60mm,-41mm);
\node at(30mm,-50mm){Initially allocated cake};
\end{tikzpicture}
\caption{Example of a trim. Agents $1$ and $2$ trim their most preferred piece (the left most piece) to the value equal to that of their second most preferred piece. 
In this instance, agent $2$ trims more than agent $1$, hence his trim is to the right of agent $1$'s trim. 
Let us assume agent $1$ and $3$ each get a complete piece with $1$ getting his second most preferred piece.   
Agent $2$ is not envious of other agents if he gets the right side of the trimmed piece up till his trim. 
If agent $2$ gets the part to the right of agent $1$'s trim instead of his own trim, then he is even happier. 
Agent $1$ is still not envious of agent $2$ if $2$ gets $\gamma$. The reason is that agent $1$ thinks that his advantage over agent $2$ is of value that is equal to his value for $\gamma$. }
\label{fig:trim}
\end{figure}
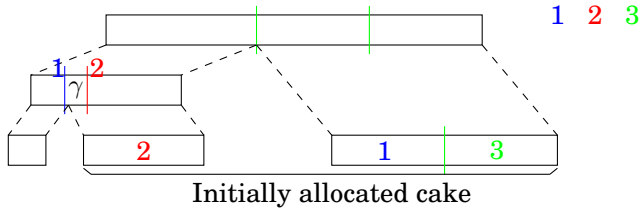

\section{The Protocol}

The overall protocol we present is essentially the Main Protocol that calls other protocols. Some protocols such as the SubCore Protocol and the Main Protocol recursively call themselves. 
The overall idea of the protocol is as follows. 
We make an agent (the cutter) divide the cake into $n$ equally preferred pieces. The cake is then allocated to the agents with each agent getting a part of one of the pieces. The subroutine will be referred to as the\emph{ Core Protocol}. The Core Protocol itself relies on recursively applying the \emph{SubCore Protocol}. When running the Core Protocol, if all the cake is already allocated, we are done. Otherwise, the Core Protocol is run repeatedly on the left-over cake. If some cake is still unallocated, we repeat it with another agent as cutter. 
After the Core Protocol has been run repeatedly, it is checked whether agents have discrepancy on a part of the cake. Discrepancy implies that agents have radically different valuations over two segments of the cake. This works to our advantage because we can run the \emph{Discrepancy Protocol} in which we let the differing agents concentrate on the segments they prefer much more which reduces our problem to envy-free cake cutting for a smaller number of agents. 
Otherwise, we exchange some parts of the cake that certain agents hold using the GoLeft Protocol in a systematic way to obtain further dominations between agents. 
After GoLeft and the subsequent Core calls, every agent outside the returned set dominates every agent inside it. 
At this point, we have decomposed the problem into a smaller problem with less number of agents. Hence the envy-free algorithm can be called recursively on the unallocated cake to divide the cake among the specified subset of agents. 
We first present the Core Protocol that is the work-horse of the overall Main Protocol. 
 
\subsection{Core Protocol}
We first present the Core Protocol that calls the SubCore Protocol. The Core Protocol is helpful in allocating additional residue in an envy-free manner. The main challenge is that just by running the Core Protocol repeatedly on the residue, there is no guarantee that the cake will be allocated completely in a bounded or even finite number of steps. Hence, we will introduce other protocols in addition to the Core protocol that help us allocate the whole cake in an envy-free manner.

\paragraph{Conventions for Core and SubCore.}
In a one-agent Main or Core call give the cake to that agent. With one
agent SubCore assigns a favorite available piece. Core returns empty
labelled shares on empty input. A cutter who values a nonempty residue at
zero can subdivide it by length. On a disconnected residue, adjacency
refers to its components in their inherited order.  All root calls
to SubCore have zero benchmarks, including the direct call in
ConnectedPieces.  A recursive call retains the symbolic history of its
caller as described in Section~\ref{tiearg}; ``forget the previous
allocation'' cancels the tentative assignment, not that history.
In line~\ref{addagent}, adding $i$ to $W$ includes tentatively assigning
to $i$ the selected available piece $a$, from its current left margin;
the following benchmark update uses this augmented tentative allocation.
For query accounting, retain each agent's physical value of every
elementary interval when a cut creates it, including between calls. Initially the table contains the value of the whole cake for each agent. A query bound for a call on a disconnected residue counts the additional queries using this table.

\begin{algorithm}[h!]
\caption{Core Protocol}
\label{algo:C}
 \begin{algorithmic}
 \scriptsize
 \REQUIRE Specified cutter (say agent $i\in N$), agent set $N$ such that $i\in N$, and unallocated cake $R$.
 \ENSURE An envy-free allocation of cake $R'\subset R$ for agents in $N$ and updated unallocated cake $R\setminus R'$.
 \end{algorithmic}
 \begin{algorithmic}[1]
  \scriptsize
  \STATE Generate the tie-breaking values that will be used to break all ties as described in Section \ref{tiearg}.
  \STATE Ask agent $i$ to cut the cake $R$ into $n$ equally preferred pieces.
  \STATE Run SubCore Protocol on the $n$ pieces with agent set $N\setminus \{i\}$ with each agent having a benchmark value as zero. 
  The call gives an allocation to the agents in $N\setminus \{i\}$ such that one of the $n$ pieces is untrimmed and unallocated.
  \STATE Give $i$ one of the unallocated untrimmed pieces from the previous step. 		
  \RETURN the partial allocation of one suffix per agent and the unallocated cake. 
 \end{algorithmic}
\end{algorithm}

\begin{algorithm}[h!]
\caption{SubCore Protocol}
\label{algo:SC}
\begin{algorithmic}
\scriptsize
\REQUIRE Cake cut into $n''$ pieces (with $n'\leq n''\leq n$) to be allocated among agents in set $\{1,\ldots, n'\}=N'\subseteq N$ with $n'=|N'|$ and a benchmark value $b_j$ for each $j\in N'$.
\COMMENT{The root benchmarks are zero. Recursive calls inherit the matched entries and benchmarks specified in Section~\ref{tiearg}.} 
\ENSURE A neat envy-free partial allocation for agents in $N'$ in which each agent $j\in N'$ gets an available suffix of symbolic value at least $b_j$.
\end{algorithmic}
 \begin{algorithmic}[1]
\scriptsize


\FOR{$m=1$ to $n'$}\label{subcore:forloop}

  \IF{the piece agent $m$ preferred at the launch of the protocol is still unallocated} 
    \STATE we tentatively give agent $m$ that piece and go to the next iteration of the `for' loop. 
  \ELSE ~ \begin{itemize}
  	\item the first $m$ agents are contesting for the same $m-1$ tentatively allocated pieces. We call them the \emph{contested} pieces. 
  	\item For each agent $j$, set $b_j'$ to be the maximum of $b_j$ and agent $j$'s value of the most preferred uncontested piece. \label{updatebench1}
	\item Then each agent in $[m]$ is asked to place a trim on all contested pieces of high enough value so that the contested piece on the right hand side of his trim is of the same value as $b_j'$\label{onlytrim}
 
  \end{itemize}
    \STATE Set $W$ to be the set of agents who trimmed most (had the rightmost trim) in some piece. \label{fwinners}
    \WHILE{$|W|<m-1$}
      \STATE Ignore the previous trims of agents in $W$ from now on and forget the previous allocation. 
      \STATE \label{step:subcore} Run SubCore Protocol on the contested pieces with $W$ as the target set of agents (with $b_j'$ as their benchmark value input) and for each contested piece, the part to the left side of the right-most trim by an agent in $[m]\setminus W$ is ignored. 
~~\COMMENT{The result of the recursive call of SubCore is an allocation that gives a (partial) contested piece to each of the agents in $W$.}

   \STATE   Take any unallocated contested piece $a$. The current left margin (beyond which the piece is ignored) is by agent $i\in [m]\setminus W$.\label{addagent}

  \[W\longleftarrow W\cup \{i\}.\]
  At this point the current allocation of agents in $W$ is tentative and not permanently made.
  \COMMENT{An agent from $m\setminus W$ has been added to $W$. For the updated $W$, each agent in $W$ gets a partial contested piece and no agent envies an unallocated piece. Recall that for each piece, the left side of the right-most trim by an agent in $[m]\setminus W$ is ignored.} 
	\STATE Update the value $b_j'$ of each agent $j$ in $W$ to equal the value of the piece that they have been tentatively allocated. \label{updatebench2}
     \ENDWHILE
     \STATE Run SubCore on all agents in $W$ and the set of contested pieces, where we ignore the part to the left of the trim made by the agent in $[m]\setminus W$. The benchmark of each $j\in W$ is $b_j'$. \label{recursive}
      ~~\COMMENT{At this point $|W|=m-1$ and each agent in $W$ has a tentatively allocated contested piece}
     \STATE  The only agent $j$ remaining in $[m]\setminus W$ is tentatively given his most preferred uncontested piece. 
   \ENDIF
  \ENDFOR                
\RETURN a neat allocation of one available suffix per agent in $N'$ meeting the benchmarks,
as well as the unallocated cake.
\end{algorithmic}

\end{algorithm}

The Core Protocol asks a specified agent termed as the cutter to cut the unallocated cake into $n$ equally preferred pieces. It then calls the SubCore Protocol that allocates to each of the agents one of the pieces (possibly partially) in an envy-free manner. No agent is given cake from any other piece. 

The name of the protocol is the same as the Core Protocol used by Aziz and Mackenzie~\cite{AzMa16a} for the case of 4 agents which returns an envy-free partial allocation in which one agent cuts the cake into four equally preferred pieces and the cutter as well as at least one other agent gets one of these four pieces. The new Core Protocol can be considered as a useful generalization of Core Protocol for the four-agent algorithm of Aziz and Mackenzie~\cite{AzMa16a}.
Unlike the Core Protocol for the four-agent case, the new Core Protocol requires a more sophisticated recursive SubCore Protocol. 
The SubCore Protocol results in an allocation that satisfies some `neat' properties. 

\begin{definition}[Neat Allocation]
Consider a cake divided into $n$ pieces and $m\leq n$ agents.  An
allocation to the $m$ agents is \emph{neat} if
\begin{inparaenum}[(i)]
\item each agent's allocation is a (not necessarily whole) part of a
distinct input piece;
\item at least one input piece is unallocated if $m<n$;
\item no agent prefers an unallocated whole input piece over his allocation; and
\item no agent prefers another agent's allocation over his allocation.
\end{inparaenum}
Allocating a part, even a physically empty part, reserves its input-piece
label.  If $m=n$, the unallocated-piece conditions are vacuous.
\end{definition}

We can prove that the Core Protocol results in an envy-free partial allocation in which the cutter cuts the cake into $n$ pieces, each agent gets a part of exactly one of the pieces, at least one non-cutter agent gets a complete piece and at least one piece is unallocated which no non-cutter agent envies over his allocation (Lemma~\ref{lemma:core-neat}).
 Since two agents get full pieces, from the cutter's perspective, at least $2/n$ of the cake is allocated after an iteration of the Core Protocol. 
 
 Most of the technical work in the Core Protocol is done when it calls the SubCore Protocol. 
The SubCore Protocol takes as input agents and pieces of cake where the number of pieces is at least as much as the number of agents. Each agent is given a part of exactly one piece.  The SubCore Protocol also takes as input the benchmark value for each agent. When we call SubCore for the first time during the Core Protocol, the benchmark value of each agent is zero. 
The benchmark value of an agent indicates the minimum value an agent expects in an envy-free allocation returned by the SubCore Protocol. 

In the SubCore Protocol, the main idea is that we start from a single agent and gradually grow the number of agents in a specified order while making sure that a neat allocation exists for the growing set of agents $\{1,\ldots, m\}$ in which at least one agent gets a full piece. 
 We denote $\{1,\ldots, m\}$ by $[m]$.
 All the allocations encountered during the course of the SubCore Protocol are considered tentative until the final step. 
If the next agent $m$ in the specified order most prefers a piece that is not currently (tentatively) allocated, we can easily handle the new agent as he can get the unallocated piece without causing envy for anyone. Otherwise, we have to do more work to ensure that the previously handled agents as well as the new agent can simultaneously get a neat allocation in which at least one agent gets a full piece. This includes calling the SubCore Protocol recursively for a smaller number of agents.

In the SubCore Protocol, if the $m$-th agent is also interested in one of the $m-1$ tentatively assigned pieces, we refer to the $m-1$ allocated pieces as the contested pieces. 
Each agent in $[m]$ is asked to trim pieces among the $m-1$ contested pieces that are of higher value than his most preferred piece outside the $m-1$ allocated pieces so that the value of the former is the same as the value of the latter. 
The value is referred to as the benchmark value of the agent. 
The rationale for asking all the $m$ agents to trim up to their benchmark value is as follows. 
If $m-1$ agents each get a part of a contested piece, then one of the $m$ agents who does not get a contested piece can get a most preferred uncontested piece. 
By asking all agents to place trims according to their benchmark values, the agent $j$ who is forced to get an uncontested piece will not be envious of other agents because each other agent gets a piece to the right of $j$'s trim on that piece. 
Note that in order to set the benchmark value of an agent we not only consider the uncontested pieces in that call of the SubCore but also implicitly consider other pieces that were part of the input of the root call to the SubCore. This ensures that agents are not envious or interested in pieces that are unallocated. 
The set of agents who have the rightmost trim in some piece is referred to as $W$. Note that $W$ is the set of agents who are guaranteed to have an envy-free allocation where each agent gets a piece that he trimmed most. If $|W|=m-1$, then if agents are given the piece they win up till their trim, then the corresponding allocation is envy-free for agents in $W$ and gives each agent in $W$ their benchmark value. Nonetheless, we call SubCore recursively on $W$ and the contested pieces with the left aside of trim the agent in $[m]\setminus W$ ignored. Doing this ensures that agents in $W$ get as much of the contested pieces as possible without causing envy.
The remaining agent in $[m]\setminus W$ is considered to be `kicked out' of the contested pieces and has to make do with a piece outside the contested pieces. He is given the most preferred uncontested piece. 

We may not be lucky and the number of agents who win some piece may not be $m-1$, i.e., some agent may have the rightmost trim in multiple contested pieces. 
The protocol then increases the size of $W$ one by one. $W$ is expanded as follows. The previous trims of agents in $W$ are ignored. The  SubCore Protocol is called recursively with $W$ as the  target set of agents and for each piece, the left side of the right-most trim by an agent in $[m]\setminus W$ is ignored. Note that by ignoring the trims of agents in $W$, we have more cake that could potentially be allocated than the allocation where the trims of agents in $W$ are not ignored but each agent in $W$ gets the right hand side of the piece where he had the rightmost trim.
In case $|W|<m-1$, we can continue increasing the size of $W$ while ensuring that agents in $W$ can get an envy-free allocation by getting a partial piece each from among the contested pieces. Take any unallocated contested piece $a$ for which the current left margin (beyond which the piece is ignored) is by agent $i\in [m]\setminus W$. We can add such an $i$ to $W$.
As we increase the size of $W$, we are able to clear more space for some piece since the non-winner because of which we were ignoring the left part of some piece is now a winner, so we an ignore his previous trim.
When $|W|=m-1$, we have ensured that all the contested pieces have been partially allocated to one agent each from $W$. In that case the remaining agent in $[m]\setminus W$ can then be given the most preferred uncontested piece.

    Note that for $n=2$, the Core Protocol coincides with the well-known Divide and Choose protocol. Next, we further explain the SubCore Protocol with the help of an example. 

   \begin{example}[SubCore Protocol]
   In order to illustrate how the SubCore Protocol works, let us explain the SubCore Protocol when $n'=3$ and there are $n=4$ pieces $a,b,c,d$. 

   In the for loop, when $m=2$, if agent $2$ has the same most preferred piece as agent $1$'s tentative piece $a$, then piece $a$ is the contested piece.
Both agents place trim on that piece to equal the value of the second most preferred piece. Let us assume that 2's trim is to the right of 1's trim.
In that case the cake in $a$ to the left of $1$'s trim is ignored temporarily as if the cake to the left did not exist. 
Agent $1$ can take a full unallocated piece $b$ that is at least as preferred by him as $2$'s tentative piece. Agent 1 is considered being kicked out of his piece and gets a piece of exactly his benchmark value. 
Since there is exactly one contested piece and the piece is `won' by agent $2$, we call SubCore on piece $a$ and agent $2$ with the cake to the left of $a$'s trim ignored. Hence agent $2$ gets $a$ up till $1$'s trim.

 We now increment $m$ to handle agent $3$. If agent $3$'s most preferred piece among the four pieces is not one currently held (even partially) by 1 and 2, then we can simply give 3 that piece. Otherwise, all agents $1,2,3$ are asked to place a trim over the contested pieces among $\{a,b\}$ to make the value equal to the most preferred piece among $c$ and $d$.
Let us assume that $3$ places the rightmost trim over $b$ and over $a$.
In that case, $W=\{3\}$.
 The fact that $3$ is in $W$ means that there exists a neat allocation for the agent in $W=\{3\}$ where he can get one of the pieces which are contested.
 We now ignore the trims of agent $3$ on pieces $a$ and $b$ and ignore the left side of the rightmost trim by agents in $\{1,2\}$.
 We call SubCore Protocol on the pieces $a,b$ with the left side of $1$ and $2$'s trims ignored and with agent $W=\{3\}$ as input. It returns a neat allocation (with respect to the current left margins on pieces beyond which they are ignored) for agent $3$ where $3$ will be assigned one of the pieces $a$ or $b$ up to the non-winner (not in $W$) agent's rightmost trim. Let us say $3$ gets piece $a$. Since the allocation is neat (with respect to the current left margins), $3$ does not envy any of the other pieces up to the non-winner's trims, meaning a new agent can be allocated a contested piece. Let us say this agent is agent $1$, and that $1$ is now the rightmost trimmer of piece $b$ and his trim coincided with the current left margin of $b$ beyond which $b$ is ignored. $1$ is added to $W$ and $1$ is allocated $b$ up till the current left margin. 
Since now $|W|=m-1$, we can make one last recursive call for agents $1$ and $3$ with the left part of the trims by agent 2 is ignored. 
Suppose this call allocates piece $a$ to $3$ and piece $b$ to $1$. The uncontested pieces $c$ and $d$ are untrimmed. Agent $2$ can be seen as having been `kicked out' of the contest, and therefore is left with his most preferred uncontested `benchmark' piece, let us say $c$. Piece $d$ has remained unallocated and untrimmed. At this point the allocation of the SubCore Protocol is returned.
   \end{example}

\subsubsection{Dealing with ties}\label{tiearg}

An example of \citet{BB17} identifies a tie in an
earlier version of the argument for admitting another winner.
When two trims coincide physically, SubCore still needs to know whose
trim is rightmost. We distinguish such trims using successively smaller
formal infinitesimals. We also use them when a recursive call forgets
an allocation but retains its benchmark values. Each agent has his own
secondary values. Only the real components determine physical cut targets.

During a root call to SubCore keep, for each input piece, an ordered list
of the suffixes encountered, from longest to shortest. An entry records
the piece label and left endpoint. Rightmost means last in this common list. Several entries may have the same
endpoint. Agent $i$ values entry $S$ by a vector
\[
   \widehat V_i(S)=(V_i(S),q_{i,1}(S),\ldots,q_{i,d}(S))
   \in\mathbb R\times\mathbb Q^d
\]
in lexicographic order. The endpoints are nondecreasing along the list,
and every agent's vectors strictly decrease. Initially the whole and empty
suffixes have secondary coordinates $1$ and $0$ respectively, even when
the physical piece is empty. A zero benchmark is $(0,0)$. Recursive
calls restrict the available lists at the caller's margins. Ignoring a
mark or restoring a prefix leaves all entries in place; a call keeps its
own marks when its descendants add theirs. Fix agent and piece orders
for resolving favorite ties, choosing among a winner's labels, and
ordering trims within a stage. The proof permits every choice of an
unallocated piece in line~\ref{addagent}.

Every refinement appends a least-significant coordinate. Its default
value is zero on every old entry and every stored numerical benchmark.
Thus every old strict comparison persists. We use two refinements.

\paragraph{Placing a trim}
Agent $i$ must mark an available piece at target $b>0$ no greater than
his value of its available whole. Append a coordinate. If an available
entry equals $b$ before the extension, give its new coordinate value
$1$ for $i$ alone. There are now consecutive available entries $U,L$
with $\widehat V_i(U)>b>\widehat V_i(L)$, since the final empty entry is
worth zero. Insert a fresh mark $Q$ between them with $\widehat V_i(Q)=b$.
Writing $\varphi$ for the real component of $b$, use the endpoint of
$U$ if $V_i(U)=\varphi$, and otherwise that of $L$ if $V_i(L)=\varphi$.
In the remaining case cut away $V_i(U)-\varphi$ from the left of $U$.
Both real inequalities are then strict, so every admissible Cut reply
is strictly between the two endpoints.

For each other agent $h$, give $Q$ secondary coordinates that put its
vector strictly between his vectors at $U,L$. If the three real values
agree, use the componentwise midpoint of the secondary vectors. If the
real value at $Q$ agrees with that at $U$ alone, copy $U$'s secondary
coordinates and subtract $1$ from the last; for agreement with $L$ alone,
copy $L$'s and add $1$ to the last. Otherwise use zeros.

\begin{claim}\label{conserv}
Each trim has a distinct position in its piece's list, preserves its
strict ordering for all agents, and costs at most one ordinary Cut and
$O(n)$ Evaluate queries. Symbolic envy-freeness implies physical
envy-freeness.
\end{claim}
\begin{proof}
The inequalities at the two neighbors verify the list order. Value the
new elementary intervals by evaluating one child for every agent of the
surrounding Core or ConnectedPieces instance, including the cutter if
present, and subtracting from the recorded parent value. On a residue
of several ordered intervals their recorded
values locate the component containing the cut: use an existing boundary
if its cumulative value is the target, and otherwise cut inside the
selected component. All secondary operations
are internal. Lexicographic comparison gives the weak inequality on
real values.
\end{proof}

\paragraph{Recording an allocation before it is forgotten}
Given an envy-free tentative allocation of entries $T_i$ on distinct
piece labels, append one coordinate. For each holder $i$, give $T_i$
the value $+1$ for $i$ and $-1$ for each other holder. All other
record--agent pairs receive zero. This preserves the strict list orders.
If $b_i=\widehat V_i(T_i)$ was copied before the refinement, then
\begin{equation}\label{eq:guard-invariant}
   \widehat V_i(T_i)>b_i,\qquad \widehat V_h(T_i)<b_h\quad(h\ne i)
\end{equation}
for the holders, since the old allocation was envy-free. These strict
comparisons persist through all later refinements.

Apply this rule immediately before a contested insertion round to the
existing holders and their entries. Each holder thereby strictly prefers
his old allocation to the others' allocations and to every unallocated
whole input piece. Apply the same rule at entry to a recursive call,
before computing favorites. For the first call on a winner set, take
one rightmost mark on a won piece for each winner. For each subsequent
call use the tentative allocation whose values have just been copied
to the benchmarks. Lemma~\ref{lemma:core-neat} verifies that these
entries form the required allocation. They are available in the caller's
history and require no queries.

\begin{claim}\label{noties}
During a contest a non-winner's retained mark is worth exactly his
retained benchmark. An agent's favorite at launch is still a favorite
when he is first processed.
\end{claim}
\begin{proof}
At the initial trim stage an agent marks each contested piece at most
once. His equality refinement on one piece leaves his marks on other
pieces and his uncontested alternatives unchanged. Later recursive calls
affect old entries only for their participants, all of whom are winners.
Thus the non-winner's mark and benchmark acquire only zeros until he
joins the winner set. Likewise the old input entries for an agent still
awaiting his turn acquire only zeros after launch.
\end{proof}

At the end of the root call discard the secondary coordinates. Suffixes
of distinct labels have disjoint physical interiors even if some of them
are physically empty.

\begin{definition}[Snapshots]
When we run the Core Protocol we end up with an envy-free partial allocation of the cake and a residue. We call the partial allocation a \emph{snapshot} or a \emph{core snapshot}. 
\end{definition}
The algorithm will keep track of certain snapshots which we will label $p_j$ with $j \in \{1, \ldots, C'\}$. The pieces allocated in each snapshot $p_j$ are labelled $c_{jk}$ where $k \in \{1, \ldots, n\}$ indicates which agent got the piece.
Since the allocation in each snapshot is envy-free, each agent thinks he got at least as much value for his piece as any other allocated piece. Therefore each agent thinks he has some (possible zero or more) bonus value over another agent in the Core snapshot. Our protocol will make use of these bonuses. 
In our protocol, we repeatedly call the Core Protocol over the resultant residue thereby making the residue smaller until the set of snapshots have enough structure that we will exploit later.  


\subsection{Groundwork for the Main Protocol}

\paragraph{Parameters}
We will use four parameters to represent the bounds we work with. These are \[C\ll C'\ll B \ll B'.\] These are constant once $n$ is fixed, however they are dependent on $n$. Since the Main Protocol calls itself on a strict subset of the agents, we will use the notation $B'_{n-1}$ to label the bound on the Main Protocol run on $n-1$ agents.
The bounds correspond to the following concepts.
\begin{itemize}
\item  $C'$ is the number of snapshots generated by the main algorithm that we label and keep track of. All subsequent partial allocations generated by runs of the Core Protocol simply serve the purpose of making the residue smaller from the cutter's perspective. Once the main algorithm has extracted all the pieces it needs from the residue and if a discrepancy has not been successfully exploited, the GoLeft Protocol will be run. 
\item When we run GoLeft we select its working set from $C$ isomorphic snapshots among the $C'$ snapshots. 
\item The value $B'$ corresponds to the total number of queries required to run the whole protocol. 
\item $B$ is used to define what we call significant pieces or values. Running the Core Protocol guarantees that a piece of cake is made smaller from the perspective of the cutter. The value $B$ is a bound on the number of times the algorithm allows us to run the Core Protocol to make the residue smaller. 
\end{itemize}

For the sake of achieving our boundedness results, the following values of the parameters work. 
\begin{inparaenum}[(i)]
\item $C=n^{n^n}$
\item $C'=n^{n^{n^n}}$
\item $B=n^{n^{n^{n^n}}}$
\item $B'=n^{n^{n^{n^{n^n}}}}$.
\end{inparaenum}
We have not optimized the values of the bounds so we expect that our general algorithmic approach works for better bounds.

In the protocol we often want to make a piece of cake (the residue) smaller. To do so we run the Core Protocol on it. To make descriptions more succinct we define a function which converts the number of times we run the Core Protocol on a piece of cake to how much smaller it is from the cutter's perspective. Note that when the Core Protocol is run, at least $2/n$ value of the cake is allocated to the agents because the cutter cuts the cake into $n$ equally preferred pieces and at least two pieces are fully allocated. By running the Core Protocol $B$ times with the same cutter, the cutter values the residue at most $f(B)={(\frac{n-2}{n})}^B$ times his value for the cake before these calls.
 

\begin{definition}[Bound function $f$]
Let   $f(B)={(\frac{n-2}{n})}^B$.
%
\end{definition}


For any piece of cake $c$ in a given snapshot, each agent has a bonus value which corresponds to how much more cake he thinks he got in that snapshot than he would have got had he been allocated $c$ instead.
\begin{definition}[Bonus value]
An agent $i$'s \emph{bonus value} on piece of cake $c_{jk}$ in snapshot $p_j$ is the value $V_i(c_{ji})-V_i(c_{jk})$, where $c_{ji}$ is the piece that was allocated to agent $i$ in $p_j$ by the Core Protocol.
\end{definition}

In a Core snapshot, a bonus value is \emph{significant} for an agent if we can make the residue smaller than that value from the agent's perspective in a bounded number of steps. Note that the significance of a piece of cake does not depend on the absolute value that an agent ascribes to it but it is relative to the value of the unallocated cake (residue).

\begin{definition}[Significant value]
An agent $i$ thinks a value is \emph{significant} if the value is more than or equal to $V_i(R)f(B)$ where $R$ is the unallocated residue. A piece is significant for an agent if it has significant value for him.
\end{definition}

Note that if an agent $i$ finds a piece of cake significant, he will still find it significant if the residue becomes smaller than before. We also observe the following about the Core Protocol and the cutter getting a significant advantage over at least one other agent.

\begin{remark}\label{remark:cutteradv}
For $n\geq 3$, after a Core call with cutter $i$ and residue $R$, the cutter has a bonus of at least $V_i(R)/(n-2)$ over another agent.
Indeed, at least two input pieces are allocated in full and the remainder of $R$ comes from at most $n-2$ pieces. One of their holders received at least $V_i(R)/(n-2)$ less than the cutter, unless $V_i(R)=0$, in which case the cutter already dominates every other agent.
After $k$ further calls with cutter $i$, the new residue is worth at most $V_i(R)((n-2)/n)^k$ to him. Thus $k=\lceil\log(n-2)/\log(n/(n-2))\rceil$ suffices for domination. For two agents both input pieces are allocated in full.

\end{remark}

    \begin{figure}
    \centering
         \scalebox{1}{
    \begin{tikzpicture}[every node/.style={draw=none,fill=none,font={\small}}]
    \coordinate (0b) at (0mm,-5mm); 
    \coordinate (0t) at (  0mm, 5mm);
    \coordinate (1b)  at ($(0b)  - (50mm, 0)$); 
    \coordinate (1t)  at ($(0t)  - (50mm, 0)$);
    \coordinate (2b)  at ($(1b)  - (10mm, 0)$); 
    \coordinate (2t)  at ($(1t)  - (10mm, 0)$);
    \coordinate (3b)  at ($(2b)  - (10mm, 0)$); 
    \coordinate (3t)  at ($(2t)  - (10mm, 0)$);
    \coordinate (kb)  at ($(3b)  - (16mm, 0)$); 
    \coordinate (kt)  at ($(3t)  - (16mm, 0)$);
    \coordinate (k1b) at ($(kb)  - (10mm, 0)$);
     \coordinate (k1t) at ($(kt)  - (10mm, 0)$);
    \coordinate (k2b) at ($(k1b) - (10mm, 0)$); 
    \coordinate (k2t) at ($(k1t) - (10mm, 0)$);
    \coordinate (lb)  at ($(k2b) - (16mm, 0)$); 
    \coordinate (lt)  at ($(k2t) - (16mm, 0)$);
    \coordinate (l1b) at ($(lb)  - (10mm, 0)$); 
    \coordinate (l1t) at ($(lt)  - (10mm, 0)$);
    \coordinate (n1b) at ($(l1b) - (16mm, 0)$); 
    \coordinate (n1t) at ($(l1t) - (16mm, 0)$);
    \coordinate (nb)  at ($(n1b) - (10mm, 0)$); 
    \coordinate (nt)  at ($(n1t) - (10mm, 0)$);
    \draw (1b) -- (1t); \draw (2b)  --  (2t); \draw (3b)  --  (3t);
    \draw (kb) -- (kt); 
    \draw (k1b) -- (k1t); 

    \draw (0b) -- (k1b) -- (k1t) -- (0t) -- (0b);


    \node at ($(0t)  - (25mm, 5mm)$) {$c_{jk}$:  piece allocated to agent $k$};

    \node at (-55mm, 0mm) {$e_{jk1}$};
    \node at (-65mm, 0mm) {$e_{jk2}$};
    \node at (-90mm, 0mm) {$e_{jkl'}$};

    \node at ($(3t)  - (08mm, 4mm)$) {$\hdots$};
  
    \end{tikzpicture}
    }
    \caption{For any allocated piece in some Core snapshot, we consider extracted pieces to be attached to the piece. In the figure, piece $c_{jk}$ is allocated to agent $k$ in snapshot $p_j$. The pieces $e_{jk1}, e_{jk2},\ldots, e_{jkl'}$ are extracted from the residue and are in consideration for attachment to the piece $c_{jk}$ in the order $e_{jk1}, e_{jk2},\ldots, e_{jkl'}$.  }
    \label{fig:extractedorder}
    \end{figure}

       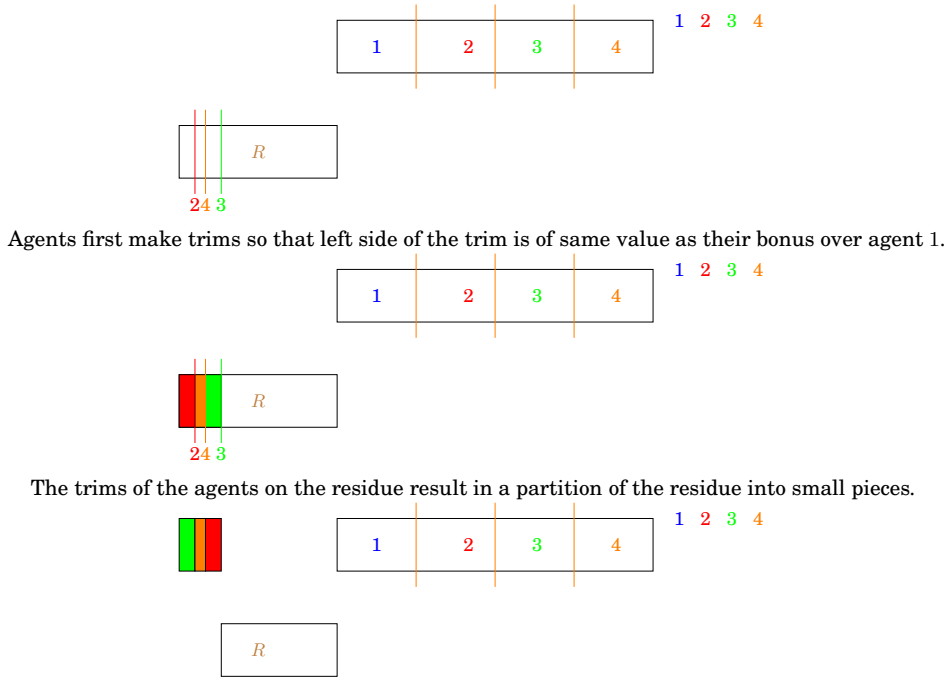
\begin{figure}[h!]

       \centering
            \scalebox{0.7}{
       \begin{tikzpicture}
       \draw (0mm, 0mm) rectangle (60mm, 10mm);
       \node[blue] at (65mm, 10mm) {1};
       \node[red] at (70mm, 10mm) {2};
           \node[green] at (75mm, 10mm) {3};
               \node[orange] at (80mm, 10mm) {4};

    \node[blue] at (7.5mm, 5mm) {1};
    \node[red] at (25mm, 5mm) {2};
     \node[green] at (38mm, 5mm) {3};
       \node[orange] at (53mm, 5mm) {4};

       \draw[orange] (15mm, 13mm) -- (15mm, -3mm);
       \draw[orange] (30mm, 13mm) -- (30mm, -3mm);
           \draw[orange] (45mm, 13mm) -- (45mm, -3mm);
    
       

        \node[brown] at (-15mm, -15mm) {$R$};
	
       \draw[] (-30mm,-20mm) rectangle (0mm,-10mm);
            \draw[red] (-27mm, -23mm) -- (-27mm, -7mm);
                \draw[orange] (-25mm, -23mm) -- (-25mm, -7mm);
                      \draw[green] (-22mm, -23mm) -- (-22mm, -7mm);
                      
                    \node[red] at (-27mm, -25mm) {2};
                    \node[orange] at (-25mm, -25mm) {4};
                     \node[green] at (-22mm, -25mm) {3};

    %
    %
    %

       \end{tikzpicture}
       }
       
         \text{\footnotesize Agents first make trims so that left side of the trim is of same value as their bonus over agent $1$.}
     
    \label{fig-extract1}

           \centering
                \scalebox{0.7}{
           \begin{tikzpicture}
           \draw (0mm, 0mm) rectangle (60mm, 10mm);
           \node[blue] at (65mm, 10mm) {1};
           \node[red] at (70mm, 10mm) {2};
               \node[green] at (75mm, 10mm) {3};
                   \node[orange] at (80mm, 10mm) {4};

        \node[blue] at (7.5mm, 5mm) {1};
        \node[red] at (25mm, 5mm) {2};
         \node[green] at (38mm, 5mm) {3};
           \node[orange] at (53mm, 5mm) {4};

           \draw[orange] (15mm, 13mm) -- (15mm, -3mm);
           \draw[orange] (30mm, 13mm) -- (30mm, -3mm);
               \draw[orange] (45mm, 13mm) -- (45mm, -3mm);
    
       
	
           \draw[] (-30mm,-20mm) rectangle (0mm,-10mm);
              \draw[fill=red] (-30mm, -20mm) rectangle (-27mm, -10mm);
               \draw[fill=orange] (-27mm, -20mm) rectangle (-25mm, -10mm);
               \draw[fill=green] (-25mm, -20mm) rectangle (-22mm, -10mm);

                \draw[red] (-27mm, -23mm) -- (-27mm, -7mm);
                    \draw[orange] (-25mm, -23mm) -- (-25mm, -7mm);
                          \draw[green] (-22mm, -23mm) -- (-22mm, -7mm);
                                  \node[brown] at (-15mm, -15mm) {$R$};

        \node[red] at (-27mm, -25mm) {2};
        \node[orange] at (-25mm, -25mm) {4};
         \node[green] at (-22mm, -25mm) {3};
           \end{tikzpicture}
           }
        
            \text{\footnotesize The trims of the agents on the residue result in a partition of the residue into small pieces.}
            \label{fig-extract2}

                                  \centering
                                       \scalebox{0.7}{
                                  \begin{tikzpicture}
                                  \draw (0mm, 0mm) rectangle (60mm, 10mm);
                                  \node[blue] at (65mm, 10mm) {1};
                                  \node[red] at (70mm, 10mm) {2};
                                      \node[green] at (75mm, 10mm) {3};
                                          \node[orange] at (80mm, 10mm) {4};

                               \node[blue] at (7.5mm, 5mm) {1};
                               \node[red] at (25mm, 5mm) {2};
                                \node[green] at (38mm, 5mm) {3};
                                  \node[orange] at (53mm, 5mm) {4};

                                  \draw[orange] (15mm, 13mm) -- (15mm, -3mm);
                                  \draw[orange] (30mm, 13mm) -- (30mm, -3mm);
                                      \draw[orange] (45mm, 13mm) -- (45mm, -3mm);
    
       
	
   
                                     
                                     \draw[fill=green] (-30mm, 0mm) rectangle (-27mm, 10mm);
                                      \draw[fill=orange] (-27mm, 0mm) rectangle (-25mm, 10mm);
                                      \draw[fill=red] (-25mm, 0mm) rectangle (-22mm, 10mm);

                                          \draw[] (-22mm,-20mm) rectangle (0mm,-10mm);
                                          
                                             \node[brown] at (-15mm, -15mm) {$R$};

                                  \end{tikzpicture}
                                  }
                           \text{\footnotesize The pieces resulting from the trims are extracted from the residue and associated with the corresponding piece.}
                           


           \caption{Process of extraction for piece allocated to agent 1 in the Core snapshot}
       \label{fig-extract} 
       \end{figure}

If we run the Core Protocol repeatedly with the same agent as cutter and we are lucky that the cutter has a significant advantage over each agent in some snapshot, then in $B$ extra iterations of the Core Protocol with the same cutter, we can make the cutter dominate each agent. This simplifies our problem, because we now only need to allocate the remaining cake among the dominated agents in an envy-free manner. In general, we may not be so lucky that the cutter dominates all other agent due to which we have to do more work to ensure that a set of agents dominates the other agents. In this additional work, the process of extraction is crucial:

\paragraph{Extraction}

For each piece allocated in each of the Core snapshots and for each agent, the Main Protocol tries to associate a piece of the cake extracted from the residue.
For the $C'$ snapshots, there are $C'n$ pieces allocated to the agents. For each of the pieces, other agents who do not get the piece \emph{extract} corresponding pieces from the residue that is the unallocated cake. Take for example, a piece $c$ allocated to some agent (say agent $1$) in a snapshot. In the snapshot in which $c$ is allocated to $1$, each agent $i$ gets a piece that is of value (according to $i$) at least as much as $c$. In the residue, each agent $i$ is asked to put a trim so that the cake from the left extreme of the residue to the trim is of value equal to $i$'s value of his piece minus his value for $c$. The trims of the agents on the residue give rise to pieces of cake corresponding to the agents' bonus value over piece $c$. These pieces are \emph{extracted} (cut away from the residue) and  \emph{associated with $c$} (kept in consideration with piece $c$). When agents make trim marks on the residue in accordance with their bonus (over a piece in a snapshot), each pair of successive trim marks gives rise to a separate piece that can be extracted. For piece $c_{jk}$ in snapshot $p_j$, we will denote by $e_{jk1}, e_{jk2}, \ldots, e_{jkl'}$ the set of pieces that are extracted in the same order with $e_{jk1}$ extracted first (see Figure~\ref{fig:extractedorder}). We say that a piece $e_{jkl}$ is extracted by agent $i$ if the right hand extreme of the piece coincided with the trim of agent $i$ on the residue.
Note that $l'\leq n-1$ because for each allocated piece in a snapshot, at most $n-1$ other agents can put trim marks on the residue so as to obtain at most $n-1$ extracted pieces. Each extracted piece has a clear corresponding allocated piece in a Core snapshot with which it is associated. Note that it can be the case that not all $n-1$ agents place a trim we allow extraction only for agents whose bonus values are not significant.

In our protocol, pieces are extracted from the residue only if each agent finds the pieces \emph{not} significant. The reason is that we want to keep sufficient residue to do further extractions on the residue as well as maintain structure. 

\begin{example}
Figure~\ref{fig-extract} illustrates the process of extraction and association. We focus on a single Core snapshot in which each of the four agents are allocated a piece. Since the Core allocation is envy-free, each of agents $2 ,3, 4$ think they got at least as much value as the piece that 1 got. Now agents $2 ,3,$ and $4$ are asked to place a trim mark each on the residue to indicate that the residue to the left of their trim mark is equal to their advantage over agent 1. The trim marks result in three slices of cake that are extracted. Among the extractions, since 2 has the leftmost trim, his piece is extracted first, then $4$ and then $3$. The extracted pieces are associated with $1$'s allocated piece in the snapshot. In case the pieces will be added to $1$'s piece, the piece extracted by agent $2$ will be attached first, then the piece by agent $4$ and then $3$.
\end{example}

\begin{figure}[h!]
\centering
\begin{tikzpicture}[yscale=0.5]
\draw (0mm, 0mm) rectangle (50mm, 10mm);
\node[blue] at (60mm, 10mm) {1};
\node[red] at (65mm, 10mm) {2};
\node[brown] at (10mm, 5mm) {$R$};

\draw[green] (20mm, 13mm) -- (20mm, -3mm);
\draw[green] (35mm, 13mm) -- (35mm, -3mm);
	
[fill,red] (0mm, 5mm) rectangle (19.8mm, 10mm);
[fill,blue] (0mm, 0mm) rectangle (19.8mm, 5mm);
\draw[dashed](0mm,0mm)--(-10mm,-10mm);
\draw[dashed](20mm,0mm)--(10mm,-10mm);
\draw (-10mm,-20mm) rectangle (10mm,-10mm);
	
\draw[green] (0mm, -7mm) -- (0mm, -23mm);
\draw[green] (5mm, -7mm) -- (5mm, -23mm);
	
\draw[dashed](0mm,-20mm)--(4mm,-30mm);
\draw[dashed](5mm,-20mm)--(10mm,-30mm);
\draw[dashed](5mm,-20mm)--(20mm,-30mm);
\draw[dashed](10mm,-20mm)--(25mm,-30mm);
\draw (4mm,-40mm) rectangle (10mm,-30mm);
\draw (20mm,-40mm) rectangle (25mm,-30mm);

\draw[dashed](0mm,-20mm)--(-10mm,-30mm);
\draw[dashed](-10mm,-20mm)--(-20mm,-30mm);
\draw (-20mm,-40mm) rectangle (-10mm,-30mm);

\node[brown] at (-5mm, -15mm) {$R$};
\node[red] at (22.5mm, -35mm) {2};
\node[blue] at (6.5mm, -35mm) {1};
	
\draw[dashed](20mm,0mm)--(30mm,-10mm);
\draw[dashed](35mm,0mm)--(65mm,-10mm);
\draw[dashed](50mm,0mm)--(80mm,-10mm);
\draw[dashed](35mm,0mm)--(45mm,-10mm);
\draw (30mm,-20mm) rectangle (45mm,-10mm);
\draw (65mm,-20mm) rectangle (80mm,-10mm);
		
\node[brown] at (-18.5mm, -35mm) {$R$};

\node[blue] at (37mm, -15mm) {1};
\node[red] at (72mm, -15mm) {2};
	
\draw[green] (-13mm, -27mm) -- (-13mm, -43mm);
\draw[green] (-17mm, -27mm) -- (-17mm, -43mm);
\draw[dashed](-13mm,-40mm)--(-7mm,-50mm);
\draw[dashed](-17mm,-40mm)--(-11mm,-50mm);
\draw[dashed](-13mm,-40mm)--(0mm,-50mm);
\draw[dashed](-10mm,-40mm)--(3mm,-50mm);
\draw[dashed](-17mm,-40mm)--(-20mm,-50mm);
\draw[dashed](-20mm,-40mm)--(-32mm,-50mm);
\draw (-11mm,-60mm) rectangle (-7mm,-50mm);
\draw (0mm,-60mm) rectangle (3mm,-50mm);
\draw (-28mm,-80mm) rectangle (-40mm,-70mm);
\node[brown] at (-30mm, -75mm) {$R$};
\node[red] at (1.5mm, -55mm) {2};
\node[blue] at (-9mm, -55mm) {1};

\node[brown] at (-26mm, -55mm) {$R$};

\draw (-20mm,-60mm) rectangle (-32mm,-50mm);

\draw[brown,fill] (-39mm,-80mm) rectangle (-40mm,-70mm);
\draw[pink,fill] (-39mm,-80mm) rectangle (-38mm,-70mm);
\draw[green,fill] (-38mm,-80mm) rectangle (-37mm,-70mm);
\draw[yellow,fill] (-36mm,-80mm) rectangle (-37mm,-70mm);
\draw[orange,fill] (-36mm,-80mm) rectangle (-35mm,-70mm);

\draw[red] (-39mm, -69mm) -- (-39mm, -81mm);
\draw[blue] (-38mm, -69mm) -- (-38mm, -81mm);
\draw[red] (-37mm, -69mm) -- (-37mm, -81mm);
\draw[blue] (-36mm, -69mm) -- (-36mm, -81mm);
\draw[blue] (-35mm, -69mm) -- (-35mm, -81mm);
\draw[red] (-32mm, -69mm) -- (-32mm, -81mm);

\draw[red] (-12mm, -49mm) -- (-12mm, -61mm);
\draw[brown,fill] (-11.9mm,-60mm) rectangle (-11mm,-50mm);

\draw[blue] (-1mm, -49mm) -- (-1mm, -61mm);
\draw[pink,fill] (-0.9mm,-60mm) rectangle (0mm,-50mm);

\draw[red] (29mm, -9mm) -- (29mm, -21mm);
\draw[fill,yellow] (29.1mm,-20mm) rectangle (30mm,-10mm);

\draw[blue] (64mm, -9mm) -- (64mm, -21mm);
\draw[fill,green] (64.1mm,-20mm) rectangle (65mm,-10mm);	

\draw[blue] (19mm, -29mm) -- (19mm, -41mm);
\draw[orange,fill] (19.1mm, -30mm) rectangle (20mm, -40mm);

\draw[red] (1mm, -29mm) -- (1mm, -41mm);

\node at (95mm, -15mm) {snapshot 1};
\node at (95mm, -35mm) {snapshot 2};
\node at (95mm, -55mm) {snapshot 3};
\node at (95mm, -75mm) {final residue};
%
\end{tikzpicture}
\caption{Snapshots with extracted pieces. The Core Protocol is called thrice resulting in three snapshots. The Core Protocol is always called on the residue---cake that has not yet been allocated. 
Since the example involves only two agents, it is not the order of trims in extraction but only whether a trim induces a significant piece or not which determines isomorphism.
Snapshot $1$ and $3$ are isomorphic since for each allocated piece in the snapshots, the other agent was able to extract a piece from the final residue to reflect his bonus/advantage. 
For the second snapshot a piece was not extracted because agent $2$'s bonus value was significant. The extracted pieces present and the order in which they were extracted is the same for snapshot $1$ and $3$, making them isomorphic.
Note that this figure is just to demonstrate isomorphism. In reality, for two agents, the Core Protocol coincides with Divide and Choose which means that there is no residue. 
}
\label{fig:isomorphic}
\end{figure}
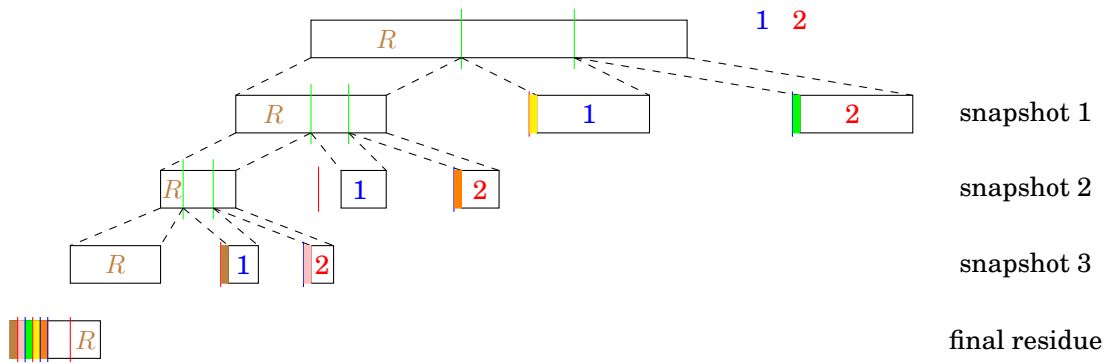

Later on, the extracted pieces may be attached to $c$ so that piece $c$ is now attractive to other agents because of the additional extracted pieces combined with $c$. 
We clarify that when pieces are extracted and associated with a given piece in a Core snapshot, such extracted pieces have not yet been  allocated to any particular agent. Extracted pieces can only be allocated after they are officially attached to their associated piece. 

For the protocol we need to restrict our focus on a subset of the snapshots where the same set of agents extracted pieces in the same order:

\begin{definition}[Isomorphic snapshots/snapshot pieces/ extracted pieces]

We call two snapshots $p_j$ and $p_{j'}$  \emph{isomorphic} to each other if for each agent $i\in N$ and for any two pieces $c_{ji}$ and $c_{j'i}$ allocated to agent $i$ in the corresponding snapshots $p_j$ and $p_{j'}$, the set of agents who extracted cake from the residue and associated to $c_{ji}$ and $c_{j'i}$ is the same and the agents in the set extracted pieces in the same order. 

We extend this notion to allocated pieces in snapshots. Two allocated pieces of cake belonging to two isomorphic snapshots are isomorphic if they were allocated to the same agent. 

We also extend the notion of isomorphism to extracted pieces. We say that for two isomorphic snapshots $p_j$ and $p_{j'}$, two extracted pieces $e_{jkl}$ and $e_{j'kl}$ are isomorphic if they are associated respectively to isomorphic pieces $c_{jk}$ and $c_{j'k}$ and that $e_{jkl}$ and $e_{j'kl}$ were extracted by the same agent.

In Figure~\ref{fig:isomorphic}, we illustrate isomorphic snapshots.
\end{definition}

Our protocol will keep track of a set of isomorphic snapshots, progressively discarding (not changing) some so that we can make manipulations on the ones we keep in an envy-free way. The pieces of cake we are working with are labelled $c_{jk}$ for allocated pieces or $e_{jkl}$ for extracted pieces. 
Suppose that the given set of snapshots is $S$.
We use the notation $c_k,S$ and $e_{kl},S$ to denote the set of pieces of cake $c_{jk}$ or $e_{jkl}$ in snapshots in $S$. Abusing the notation, we will simply say $c_k$ or $e_{kl}$ to mean the set of pieces which are in the snapshots we are working with. Note that we will generally use index $j$ for the snapshot number, $k$ for the piece number in a given snapshot, and $l$ for the $l$-th extracted piece.

\subsection{The Main Protocol}
Every subroutine commits its allocations and returns the updated residue
to its caller. In particular, the two recursive calls after a discrepancy
divide the held-out piece and the residue returned by Discrepancy.
After GoLeft, $R$ includes all the unassigned extractions it returned.

        \begin{algorithm}
        \caption{Main Protocol}
        \label{algo:M}
        \begin{algorithmic}[h]
        \scriptsize
        \REQUIRE Cake $R$ and a set of agents $N$ with $n=|N|$.
        \ENSURE An envy-free allocation that completely allocates $R$ among agents in $N$.
        \end{algorithmic}
        \begin{algorithmic}[1]
        \scriptsize
        \Algphase{Base Case}
        \IF{$|N|\leq 4$}
        \STATE Allocate the residue among the agents in $N$ using a constant-time
envy-free protocol for $|N|\leq 4$. For four agents use the protocol
of \citeN{AzMa16a}.
        \ELSE
        \Algphase{Generate Core Snapshots}
        \FOR {all $j: 1 \dots C'$} 
          \STATE For some $i\in N$ who has acted as cutter least often in this snapshot-generation loop, run Core protocol($i$, $N$, $R$); Update $R$ to the cake that is unallocated. 
        The Core Protocol allocation gives us snapshot $p_j$ and $n$ pieces of cake $c_{jk}$ with $k:1, \ldots, n$.
        \ENDFOR
	  ~\COMMENT{We have now generated $C'$ Core snapshots. After generating the first $n$ snapshots, we already have an envy-free allocation in which each agent receives at least $\nicefrac{1}{n}$ of his value for the original cake.}
          \Algphase{Extraction}
          \STATE Set $a$ to $0$.
          \WHILE{$a=0$}\label{mwhile2}
          \WHILE{for  agent $i\in N$ and some piece $c_{jk}$, \label{mfor1} $V_i(R)f(B)^2<V_i(c_{ji})-V_i(c_{jk})<V_i(R)f(B)^\frac{1}{2}$} \label{mwhile1}
           \COMMENT{\emph{We ensure that the bonus values are either larger than significant or smaller than significant by a large  factor dependent on our bound $B$}.}
            \STATE Run Core Protocol($i$, $N$, $R$).
          \ENDWHILE \label{mewhile1}
          \IF{there exists $\varnothing\ne A\subsetneq N$ such that every agent in $N\setminus A$ dominates every agent in $A$}
          Call Main Protocol($R$, $A$) to divide the unallocated cake $R$.
          \RETURN allocation of the cake to the agents.
          \ENDIF
          \STATE Set $a$ to $1$. Classify each candidate for extraction once, using the current $R$ before removing the candidate.
            \FOR{all pieces of cake $c_{jk}$} \label{mfor2}	
              \FOR{agents $i\ne k$} \label{mfor3}
                \IF{$i$'s bonus value is not significant on piece $c_{jk}$}
                  \STATE Ask $i$ to place a \textit{trim} on the residue such that the piece stretching from the left of the residue to the trim is equal to  the agent's \textit{bonus value} for that piece, which we will label $b^i_{c_{jk}}$. A zero bonus gives the left endpoint. \label{trim}
                \ENDIF
              \ENDFOR \label{mefor3}
              \STATE Label from left to right the pieces delimited by those \textit{trims} with labels $e_{jkl}$ with $l:1, \ldots, m$ where $m$ is the number of agents who placed a trim on the residue (did not have significant bonus) and $j$ and $k$ are used as previously \COMMENT{\emph{if two or more agents' trims coincide we break ties lexicographically and ascribe empty cake to the latter agents' extractions.}}
      
             \COMMENT{\emph{For a given piece of cake $c_{jk}$ the $m$ agents with non-significant bonus value have now delimited a piece of cake corresponding to that value from the residue, and those pieces have been labelled.}}
              \FOR{all $l:1, \ldots, m$} \label{mfor4}
                \IF{no agent thinks the piece $e_{jkl}$ is significant}
                  \COMMENT{ \emph{Everyone agrees piece is insignificant}}
                  \STATE Extract piece $e_{jkl}$ from the residue, associate it to piece $c_{jk}$ and update $R$ to $R-e_{jkl}$ \COMMENT{See Figure~\ref{fig-extract}}.
                \ENDIF
                \Algphase{Discrepancy}
                \IF{some agents think $e_{jkl}$ is significant but others do not (agents think the piece is discrepant)}
          
                      \COMMENT{\emph{There is a large discrepancy between agents' valuation, this causes a problem because some agents might think we are taking too much away from the residue. However discrepancy can be exploited or eliminated}.}
                  \STATE Run Discrepancy($e_{jkl}$,$b^u_{c_{jk}}$,$\{e_{jkl}\}$,$R$) with the discrepant piece $e_{jkl}$, the bonus value $b^u_{c_{jk}}$  of the agent $u$ who made the rightmost trim delimiting the piece, the set of previously extracted pieces of cake $\{e_{jkl}\}$ and the residue $R$ as input. This will return the Boolean value DISCREPANCY and when that value is $1$ the protocol also returns two sets $D \subset N$ and $D'\subset N$ such that $D$ and $D'$ partition $N$.
          
                   \label{mdisc1}
                \IF{DISCREPANCY$=1$}
                  \COMMENT{\emph{The discrepancy can be exploited}}
                  \STATE Main Protocol($e_{jkl}$, $D$). \label{msub1}
                            \STATE Main Protocol($R$, $D'$). \label{msub2}
                  \RETURN allocation of the cake to the agents.
                \ELSE
                   \COMMENT{\emph{The discrepancy cannot be exploited but $b^u_{c_{jk}}$ is now a significant bonus value.}}
                  \STATE Add $e_{jkl}$ to the returned $R$, clear all extraction marks and associations, set $a=0$, and restart the outer while loop (line~\ref{mwhile2}) immediately.
                \ENDIF
              \ENDIF
            \ENDFOR \label{mefor4}
          \ENDFOR \label{mefor2}
        \ENDWHILE \label{mewhile2}
        \Algphase{GoLeft}
        \STATE Run GoLeft($ N, \{p_j\}$, $\{c_{jk}\}$, $\{e_{jkl}\}$, $R$) Protocol with the set of labelled snapshots $\{p_j\}$, the set of allocated pieces of cake $\{c_{jk}\}$, the set of extracted pieces $\{e_{jkl}\}$ and the residue $R$ as input. Use its returned residue $R$ and set $A\subset N$. \label{mgoleft1}
        \FOR{all agents $i$} \label{mfor5}
          \STATE Run Core Protocol($i$,$N$,$R$) $2B$ times each time on the updated smaller residue $R$. 
        \ENDFOR \label{mefor5}

        \COMMENT{\emph{The previous for loop converts significant bonus values into dominance with respect to the residue $R$.}}
            \Algphase{Recursion of the Main Protocol}
        \STATE  Call Main Protocol($R$, $A\subset N$), that is the Main Protocol on a subset of agents $A \subset N$ output by the GoLeft Protocol  and with $R$ as the input cake. \label{msub3}
        \ENDIF
        \RETURN allocation of the cake to the agents.
        \end{algorithmic}
        \end{algorithm}


The Main Protocol is the engine which runs our overall envy-free protocol. It is responsible for allocating the whole cake among all the agents in an envy-free manner. The Main Protocol works recursively. If the number of agents is four or less, a previously known bounded envy-free algorithm can directly be called.  Otherwise, the Main Protocol divides part of the cake in an envy-free manner and identifies a set of agents $N\setminus A$ that all dominate agents in $A$ with respect to the cake that is unallocated. The Main Protocol then recursively calls itself to divide the remaining cake among agents in $A$. 

If the number of agents is more than four, the Main Protocol (Algorithm~\ref{algo:M}) calls the Core Protocol sequentially on the updated residue so that the residue becomes smaller than before after each call of the Core Protocol. The repeated calls of the Core Protocol help generate a number of snapshots each containing $n$ pieces of cake. In each snapshot, each agent has been allocated a piece of cake. 
Envy-freeness is maintained throughout for the allocated cake, and in the Core Protocol each agent thinks he got the highest value piece. 

\begin{remark}\label{remark:waste}
When we call the Core Protocol $n$ times each time with a different cutter,
we argue that we get an envy-free allocation in which agents get at least  $\nicefrac{1}{n}$ value of the original cake.

Suppose agent $i$ is the cutter in the $j$-th call, where $j\leq n$. Then in the first $j-1$ calls of the Core Protocol, $i$ gets at least $1/n$ of his value of the allocated cake. This follows from the envy-freeness of the allocated cake. For the remaining cake, $i$ again gets $1/n$ value of the unallocated cake in the $j$-th call because $i$ is the cutter. 
Hence $i$ gets at least $1/n$ of his value of the original cake. Since the cake is always allocated via Core, it is allocated in an envy-free manner. 
Hence for $n\geq5$, after Main has made its first $n$ Core calls, we have an envy-free allocation in which agents get at least  $\nicefrac{1}{n}$ value of the original cake.
\end{remark}

After the first $n$ calls to the Core Protocol, the Main Protocol does not stop making calls to the Core Protocol if there is still some unallocated cake. After calling the Core Protocol $C'$ times, $C'$ Core snapshots are obtained.
The Core Protocol may be further called (in the while loop in step~\ref{mwhile1}) to make the residue even smaller. This ensures that each agent considers each piece in the first $C'$ snapshots significant or smaller than significant by a large factor dependent on our bound $B$.

For each piece of cake $c$ in the $C'$ snapshots, agents can ascribe what we refer to as a \emph{bonus value}, which corresponds to how much more value they got in that snapshot than their value for piece $c$. In the Main Protocol, for each piece of cake $c$ in the snapshots, we ask every agent other than the holder whose bonus is non-significant to make a cut on the residue equal to their bonus value. The pieces obtained from these cuts are then taken from the residue and associated (but not yet attached) to piece of cake $c$. We refer to this process as extraction. A piece is extracted only if all agents find it insignificant. The extracted pieces will potentially be attached to their associated piece of cake $c$ in the GoLeft Protocol. Most of them however will be sent back to the residue or shared amongst a subset of agents in an envy-free way. The process of attaching the extracted pieces to their associated piece in the GoLeft Protocol is designed to make that piece desirable to the agent whose bonus value was used to extract the piece from the residue. 
The Main Protocol also calls the Discrepancy Protocol in case there is some piece in consideration for extraction that some agents consider significant and others do not. The goal of the Discrepancy Protocol is to exploit any such discrepancy and to ensure that when the GoLeft Protocol is called by the Main Protocol, then there is no discrepancy in how the extracted pieces are viewed, i.e., no actually extracted piece is considered significant by some agent.


\subsection{Discrepancy Protocol}

When pieces are being extracted from the residue during the Main Protocol, it may be the case that one of the pieces $e_{jkl}$ in consideration for extraction is significant for some agent.
In that case, the piece is not extracted and the Discrepancy Protocol (Algorithm~\ref{algo:D}) is called in line~\ref{mdisc1} that either exploits or `eliminates' this discrepancy. The discrepant piece $e_{jkl}$ is kept aside from the residue.
On the other hand, all previously extracted pieces are added back to the residue. 
Since the difference between a piece that is just  above significant or just below significant can be arbitrarily small, the Core Protocol is used to create a \emph{gap} so that the discrepant piece either has value at least $V_i(R)n$ or value at most $V_i(R)/n$. 
This gap is highly useful because the goal of the Discrepancy Protocol is to either ensure that (1) everyone thinks that the discrepant piece is significant or (2) the problem of finding an envy-free allocation can be broken into two sub-problems where some agents $D$ are allocated the discrepant piece and the rest $D'$ are allocated the residue. In this, we use the fact mentioned earlier that envy-freeness implies proportionality.

\begin{algorithm}
\caption{Discrepancy Protocol}   
\label{algo:D}
\begin{algorithmic}[1]
\scriptsize
\REQUIRE Residue $R$, \textit{discrepant} piece $e_{jkl}$ , bonus value $b^u_{c_{jk}}$ on $c_{jk}$ of the agent $u$ who wanted to extract the discrepant piece (but could not extract), set of extracted pieces of cake $\{e_{jkl}\midd \text{extracted pieces}\}$ , agent set $N$.
\ENSURE Possibly modified residue $R$, Boolean value called DISCREPANCY, set of agents $D$, set of agents $D'$.
\end{algorithmic}
\begin{algorithmic}[1]
\scriptsize
\STATE Take out the discrepant piece $e_{jkl}$ from $R$; Reinsert all the \emph{extracted} pieces back into the residue, relabel the aggregate piece $R$.
\Algphase{Run Core}
\STATE Run Core Protocol($i$, $N$, $R$) for every agent $i$ $B$ times iteratively on the updated residue $R$
\WHILE{for some agent $i$, it is the case that $\frac{V_i(R)}{n}\leq V_i(e_{jkl})< V_i(R)n$} \label{dwhile1}
 
  \COMMENT{\emph{The piece might be discrepant but agents who do not classify the piece the same might have arbitrarily close valuation of the piece if they are lying next to the bound. This while loop ensures that there is a gap between agents who think the piece is significant and others}}
    \STATE Run Core Protocol($i$, $N$, $R$) $B$ times iteratively on the updated residue $R$. 
\ENDWHILE \label{dewhile1}
\Algphase{Exploit Discrepancy}
\IF{$e_{jkl}$ still has some agents consider it significant and others not}
  \STATE Set $D=\{i\in N: V_i(e_{jkl})\geq nV_i(R)\}$ and $D'=N\setminus D$.
  \COMMENT{Every member of $D'$ values $e_{jkl}$ at less than $V_i(R)/n$. A member giving both pieces value zero belongs to $D$.}

\STATE Set DISCREPANCY to $1$.
  \RETURN $R$, DISCREPANCY, $D$ and $D'$.
\ELSE 
\Algphase{Cannot Exploit Discrepancy --- but all agents now on the same page}
  \COMMENT{\emph{ This means that we cannot exploit a discrepancy but now all agents think that the bonus value of whoever made the trim delimiting the discrepant piece is significant, including agent $u$ who made the trim}}
  \STATE Set DISCREPANCY to $0$.
  
  \COMMENT{\emph{$b^u_{c_jk}$ is now considered significant by agent $u$}}
  \RETURN $R$ and DISCREPANCY.
\ENDIF
\end{algorithmic}
\end{algorithm} 

If all agents find the disputed piece significant, Main restores it and restarts extraction on the new residue. The previously extracted pieces were restored before Discrepancy ran Core. Thus all the cake outside the new residue has been allocated, and no old extraction marks survive the restart. Lemma~\ref{Discb} justifies the classification. The boundedness proof below counts the restarts.

\subsection{Groundwork for the GoLeft Protocol}

\begin{algorithm}
\caption{GoLeft Protocol}
\label{algo:GL}
\begin{algorithmic}
\scriptsize
\REQUIRE Set $N$, Set $C'$ of snapshots $p_j$, set of corresponding pieces $\{c_{jk}\}$, the set of extracted pieces $\{e_{jkl}\}$, and residue $R$.
\ENSURE The updated residue $R$ and a nonempty proper set $A\subset N$ such that the Core calls following GoLeft give every agent outside $A$ domination over every agent in $A$.
\end{algorithmic}
\begin{algorithmic}[1]
\scriptsize
\Algphase{Isomorphic Snapshots and build Permutation Graph}
\STATE Select from the $C'$ snapshots a set of size $C$ of \textit{isomorphic snapshots} ( this follows by Claim~\ref{claim:find-isomorphic} which uses a pigeon-hole argument). \label{isomorphic}
\STATE Relabel the $C$ snapshots $p_j$ with $j$ now ranging from $1$ to $C$.
\STATE Take $D=\prod_{l=1}^{n-1}[(n-l+1)(ln+1)]^n$ of these snapshots as $S$. Return all extractions outside $S$ to $R$.
\COMMENT{\emph{Out of all the snapshots that the Main Protocol generated, we will focus on isomorphic snapshots.  The isomorphic snapshots that are selected are put in a set $S$ which will shrink in size as the algorithm proceeds.}}
\STATE Build the permutation graph with nodes corresponding to the agents. 
Throughout the algorithm, we maintain two sets: (1) $T$ (set of nodes/agents such that the isomorphic pieces held by them in $S$ have not had $n-1$ attachments) and (2) 
$T'$ (set of nodes/agents such that the isomorphic pieces held by them in $S$ have had $n-1$ attachments). $T$ and $T'$ partition the nodes. Initially, (1) all nodes are placed in $T$ so that $T'$ is empty; (2) each node points to itself in the permutation graph (3) each agent owns his originally allocated set of pieces allocated in the Core snapshots in the working set $S$. Thus each node has the following associated information: agent and his allocation in snapshots in $S$. 
\COMMENT{\emph{The permutation graph will keep track of which permutation we can do next in order to make the protocol progress. If the allocated pieces associated with a node have had $n-1$ attachments, the node is moved from $T$ to $T'$.}}
\WHILE{there is a node in $T$} \label{gwhile1}
  \Algphase{Cycle and Exchange}
   \STATE Find a cycle in the permutation graph which involves a node from $T$. \label{cycle}
     \COMMENT{\emph{Since nodes in $T$ all have in-degree exactly  $1$ and each node points to each node in $T'$, such a cycle exists as proved in Claim~\ref{claim:cycle-exists}}}.
\STATE Simultaneously give each agent in the cycle the family held by the agent he points to, with its prefix through his own trim (Remark~\ref{remark:consistent}).
\STATE Move each family's incoming edges and its $T$ or $T'$ membership to its new holder simultaneously.
\STATE Take a node/agent $i$ in the cycle that is from $T$.
\Algphase{Separation - the only way the while loop exits}
\IF{there is a node in the located cycle that is from $T$ (has had less than $n-1$ attachments) but has no extracted pieces to be attached}\label{step:separation}
\STATE Choose one such family. Let $A$ be its original owner and all agents who extracted from it. Freeze all remaining assignments and return all unassigned extractions to $R$.
\RETURN $A$ and the updated residue $R$.
\ELSE ~we now know that there is an agent $i$ that was in the cycle from $T$ that holds pieces in $S$ that still have extracted pieces to be attached. We focus on this agent $i$ in the attachment phase.
\ENDIF
\Algphase{Attachment --- (to be done in a subset of the snapshots)}
  \COMMENT{\emph{The next part of the protocol will attach in a subset of the snapshots the set of isomorphic extracted pieces in $e_{k(l+1)}$ to the set of pieces $c_k$, thus making pieces in set $c_k$ desirable to the agent who extracted pieces in $e_{k(l+1)}$. The set of pieces $c_k$ and its current attachments are currently owned by agent $i$. At this point for each isomorphic piece $c_k$, associated pieces up till $i$'s trim have already been attached.}}
\STATE Index this family's original owner first, its extractors in trim order next, then the other agents in fixed order.
The holder $i$ has index $l$. Write $e_{jk(l+1)}$ for the strip of the next extractor, whose index is $l+1$.
\STATE Set $S'=S''=\emptyset$ and $s=|S|$.

  \Algphase{Attachment --- making it agreeable for agents from $l+1$ to $n$ }
  \FOR{agent $i$ ranging from $l+1$ to $n$} \label{gfor1}
  \COMMENT{\emph{Ask agent $i$ who does not (so far) find pieces in $c_k$ desirable to discard enough snapshots where being allocated $c_k$ instead of whichever piece he is holding would be most disadvantageous. In other words, he reserves advantage over agents in $\{1,\ldots, l\}$.}}
    \STATE Ask $i$ to choose $\frac{s}{n-l+1}$ snapshots from $S$ for which $i$ values the difference between his bonus value for $c_k$ and the extracted pieces currently attached to $c_k$ the most.
    \STATE Remove the chosen snapshots from $S$ and add them to $S'$.
  \ENDFOR \label{gefor1}
\STATE Freeze the actual assignments in $S'$ and return all its unassigned extractions to $R$. Set $s=|S|$.
  \COMMENT{ The pieces that were associated to a piece in a discarded snapshot but not attached will never be attached, we therefore send them back to the residue. The fact that they were not significant means that the size of the residue is hardly affected by this}
    \Algphase{Attachment --- making it agreeable for agents from $1$ to $l$ }
  \FOR{all agents $i$ ranging from $1$ to $l$} \label{gfor2}
    \COMMENT{ \emph{In this part we ensure that agents $i \in \{1, \ldots, l\}$ will not be envious of agent $l+1$ even if $l+1$ is allocated $c_k$  with the additional $e_{k(l+1)}$} as well as the intermediate attachments}.
    \STATE Ask $i$ to choose the $\frac{sn}{ln+1}$ snapshots from $S$ for which he values the piece $e_{k(l+1)}$ the most.
    \STATE Remove the chosen snapshots  from $S$ and add them to $S''$.
  \ENDFOR \label{gefor2}
\STATE Freeze the actual assignments in $S''$. Set aside the pending $e_{jk(l+1)}$ there and return its other unassigned extractions to $R$.
  \COMMENT{\emph{Again extracted pieces that will never be attached are sent back to the residue}}
  \STATE All pieces $e_{k(l+1)}$ in snapshots in $S''$  are aggregated into a piece $a$.
  \STATE Run Main Protocol($a$, $\{1, \ldots, l\}$).
   \label{gsub1}
 \STATE $S'' \leftarrow \emptyset$.
      \Algphase{Attachment --- now happening}
\STATE Attach $e_{jk(l+1)}$ to $c_{jk}$ as an offer in every still active snapshot. Pending prefixes stay with their family until an exchange or until the snapshot is frozen. \label{attach}
      \COMMENT{\emph{ This update reflects that we made a set of allocated pieces desirable to a new agent}}.
\STATE Replace the incoming edge of the selected family's holder (index $l$) by the edge from the next extractor (index $l+1$) to that holder.
   \STATE If the pieces held by $i$ have had $n-1$ attachments, delete $i$ from $T$ and place it in $T'$ and make every node point to $i$. 
  	
\ENDWHILE \label{gewhile1}
\end{algorithmic}
\end{algorithm}            

The GoLeft Protocol is the heart of our overall protocol and is crucial to allocate the cake that is still not allocated. All the other steps in the Main Protocol can be viewed as preparing the ground for the GoLeft Protocol to work. By calling the Core Protocol a sufficient number of times in the Main Protocol, enough $C'$ Core snapshots  are obtained that are helpful to identify $C$ isomorphic snapshots in the GoLeft Protocol (the initial working set consists of $D$ of them).
At the time each extracted piece was accepted, every agent considered it insignificant relative to the residue at that time.

The GoLeft Protocol (Algorithm~\ref{algo:GL}) is called by the Main Protocol. 
GoLeft identifies a nonempty proper set $A$ such that each agent outside
$A$ has a sufficient advantage over every agent in $A$. It does this by
attaching extracted pieces and exchanging holdings across a working set
of isomorphic snapshots. The working set shrinks as the protocol sets
aside cake to keep the allocation envy-free. GoLeft returns when a family
chosen on the exchange cycle has fewer than $n-1$ extractions and all
of them have been offered. On that family the original holder and the
extractors form $A$. The following $2B$ Core calls per cutter turn the
advantage into domination (Lemma~\ref{domin}).

The GoLeft protocol operates on a working set of isomorphic snapshots. The main operations of the GoLeft Protocol are to attach extracted pieces to the Core snapshot allocation and to implement exchanges in which there is a sequence of agents $a_o,a_1,\ldots, a_{k-1}$ where each agent $a_i$ in the sequence gets the pieces of agent $a_{i+1 \mod k}$.
When operations such as attachments and exchanges happen, the isomorphic snapshots change but remain isomorphic nonetheless.

Before we give an overview of the steps of the GoLeft Protocol, we establish some concepts and mathematical structures we will work with. The two key structures are the \emph{working set $S$ of isomorphic snapshots} and the 
\emph{permutation graph} that is defined with respect to the working set of isomorphic snapshots. During the GoLeft Protocol, each structure gets updated based on the information on the other structure. The algorithm makes progress when agents exchange their pieces with each other.

\paragraph{Working Set of Isomorphic Snapshots}
During the course of the GoLeft Protocol, the isomorphic snapshots in $S$ get changed in the following way: (1) some subset of $S$ is removed from $S$ (2) agents exchange their pieces and (3) extracted pieces are attached to pieces in the snapshots in $S$. 
 
 When we update $S$, we maintain isomorphism and other invariant properties as follows. 
 If an agent holds a piece $c_{jk}$, he holds the whole set $c_k$ of isomorphic pieces in set of snapshots $S$. Note that we are abusing notation here as $c_k$ refers to a set of pieces allocated to agent $k$ in the working set of snapshots $S$. If an agent holds an extracted piece $e_{jkl}$, he holds the whole set $e_{kl}$ of isomorphic extracted pieces associated with $S$ as well.  
 When we implement an exchange, we are essentially making $|S|$ exchanges -- one in each of the snapshots. These are exchanges that not only involve the pieces in the snapshots but also involve those extracted pieces that have been \emph{attached} to the pieces. When we attach an extraction to a piece in a snapshot in $S$, we simultaneously attach isomorphic extractions to the corresponding isomorphic allocated piece in each of the snapshots.

We also ensure that the extracted pieces are attached in the appropriate order so that each associated piece gets attached after the previous associated pieces have been attached. For example, if we were focussing on snapshot in Figure~\ref{fig-extract}, the extracted associated piece due to agent $3$ will get attached after the extracted associated piece due to agent $4$.  Note that if we attach an extracted piece $e_{jkl}$ to an allocated piece $c_{jk}$ (along with its previously attached extracted pieces) in a snapshot in the working set $S$, we perform a similar attachment for all such isomorphic extracted pieces in the set $e_{kl}$ to their corresponding pieces in set $c_{k}$. Also, if an agent $i$ currently holds an extracted piece, he also holds all earlier extracted pieces as well. 

We make a crucial point about a consistency condition that we enforce. 


\begin{remark}\label{remark:consistent}
An agent who holds a family receives the main piece in each active snapshot and exactly the prefix of its extractions up to his own trim. The original owner's prefix is empty. This convention also applies after all $n-1$ extractions have been offered, when any agent may receive the family. An offered portion beyond the current holder's prefix stays with the family for possible future exchanges. Freezing a snapshot fixes its holders' actual allocations and releases every unassigned portion; the quantities expressly set aside for the current sharing are released to that sharing.
\end{remark}

In the GoLeft Protocol, we operate on an increasingly small set of isomorphic snapshots. The goal is to reach a set of isomorphic snapshots that can be used to find a set of agents who all dominate other agents.
As the GoLeft Protocol proceeds, we discard snapshots to allow us to focus on others where extracted pieces from $E$ have been `attached' in an envy-free way to the allocated  pieces in $C$ to which they were associated. The set of isomorphic snapshots $S$ becomes smaller when a set of isomorphic extractions $e_{k(l+1)}$ are attached to set of isomorphic pieces $c_k$ in the snapshots $S$ with each particular extracted piece in $e_{k(l+1)}$ getting attached to its corresponding piece in $e_k$.  By \emph{discarding} a set of Core snapshots, we mean that these snapshots and their allocations are not further worked upon and their unassigned extractions are returned to the residue except for those set aside for sharing. Intuitively, the purpose of discarding snapshots will be to preserve the remaining advantages of agents over other agents. 



The GoLeft procedure gradually attaches the extracted pieces to the allocated pieces in the working set of snapshots. 
The reason we call the protocol `go left' is because we can visualize that for each piece in the set of snapshots we work with, we want to add the next extracted pieces to the isomorphic allocated pieces that are kept to the left of the isomorphic allocated pieces (see Figure~\ref{fig:goleft}). By attaching the next extracted pieces, we are `going left'. 

\begin{figure}[h!]
\centering
     \scalebox{0.7}{
\begin{tikzpicture}[every node/.style={draw=none,fill=none,font={\small}}]
\coordinate (0b) at (0mm,-5mm); 
\coordinate (0t) at (  0mm, 5mm);
\coordinate (1b)  at ($(0b)  - (50mm, 0)$); 
\coordinate (1t)  at ($(0t)  - (50mm, 0)$);
\coordinate (2b)  at ($(1b)  - (10mm, 0)$); 
\coordinate (2t)  at ($(1t)  - (10mm, 0)$);
\coordinate (3b)  at ($(2b)  - (10mm, 0)$); 
\coordinate (3t)  at ($(2t)  - (10mm, 0)$);
\coordinate (kb)  at ($(3b)  - (16mm, 0)$); 
\coordinate (kt)  at ($(3t)  - (16mm, 0)$);
\coordinate (k1b) at ($(kb)  - (10mm, 0)$);
 \coordinate (k1t) at ($(kt)  - (10mm, 0)$);
\coordinate (k2b) at ($(k1b) - (10mm, 0)$); 
\coordinate (k2t) at ($(k1t) - (10mm, 0)$);
\coordinate (lb)  at ($(k2b) - (16mm, 0)$); 
\coordinate (lt)  at ($(k2t) - (16mm, 0)$);
\coordinate (l1b) at ($(lb)  - (10mm, 0)$); 
\coordinate (l1t) at ($(lt)  - (10mm, 0)$);
\coordinate (n1b) at ($(l1b) - (16mm, 0)$); 
\coordinate (n1t) at ($(l1t) - (16mm, 0)$);
\coordinate (nb)  at ($(n1b) - (10mm, 0)$); 
\coordinate (nt)  at ($(n1t) - (10mm, 0)$);
\draw (1b) -- (1t); \draw (2b)  --  (2t); \draw (3b)  --  (3t);
\draw (kb) -- (kt); \draw (k1b) -- (k1t); \draw (k2b) -- (k2t);
\draw (lb) -- (lt); \draw (l1b) -- (l1t); \draw (n1b) -- (n1t);
\draw (0b) -- (nb) -- (nt) -- (0t) -- (0b);
\draw (0b) -- (nb) -- (nt) -- (0t) -- (0b);

\draw[shade] (lb) -- (lt) -- (l1t) -- (l1b) -- (lb);

\node at ($(0t)  - (25mm, 4.5mm)$) {$c_{j1}$:  piece allocated to Agent $1$};

\node at (-55mm, 0mm) {$e_{j1}$};
\node at (-65mm, 0mm) {$e_{j2}$};
\node at (-90mm, 0mm) {$e_{jl}$};
\node at (-100mm, 0mm) {$e_{jl+1}$};

\node at ($(3t)  - (08mm, 4mm)$) {$\hdots$};
\node at ($(k2t) - (08mm, 4mm)$) {$\hdots$};
\node at ($(l1t) - (08mm, 4mm)$) {$\hdots$};
\node[below = 2mm of  1b] {$1$};   \node[below = 2mm of  2b] {$2$};
\node[below = 2mm of  3b] {$3$};   \node[below = 2mm of  kb] {$l$};
\node[below = 2mm of k1b] {$l+1$}; \node[below = 2mm of k2b] {$l+2$};
\node[below = 2mm of  lb] {$k+1$};   
\node[below = 2mm of l1b] {$k+2$};
\node[below = 2mm of n1b] {$n-1$}; \node[below = 2mm of  nb] {$n\vphantom{1}$};
\end{tikzpicture}
}
\caption{Illustration of the GoLeft Protocol on a particular piece of cake that is originally allocated to agent $1$ and in which the pieces were extracted in this particular case by agents in set $\{2,\ldots, k+1\}$ and in order $2, \ldots, k+1$.
The piece $c_{j1}$ was originally allocated to agent $1$ in snapshot $p_j$. 
When we refer to `extractions up till agent $l$'s extracted piece we will mean extracted pieces $e_{j1},\ldots, e_{j(l-1)}$.
 GoLeft offers the pieces extracted by agents $2,3,\ldots,k+1$ in that order when it selects this family. This attachment can be viewed as going left. Agents $k+2$ to $n$ placed no trim on this family because their bonuses were significant. They are the prospective dominators. 
If this family is chosen on the cycle after all its extractions have been offered and $k<n-1$, GoLeft returns $A=\{1,\ldots,k+1\}$. After the subsequent Core calls the other agents dominate every agent in $A$.}
\label{fig:goleft}
\end{figure}
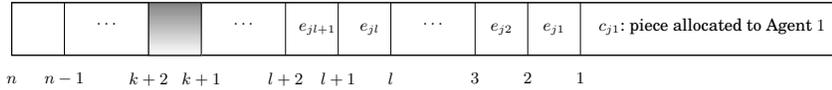



How do we know which set of isomorphic pieces in $S$ can have their next extraction? Which set of agents can exchange their currently held isomorphic pieces along with their current attachments? For this we work in tandem with the permutation graph. 

\paragraph{Permutation Graph}

The permutation graph keeps track of which agent is willing to move to which piece in the snapshots. The high level idea is that nodes of the permutation graph correspond to the agents and an agent $i$ points to another agent $j$ if he will be as happy taking $j$'s allocated pieces along with the attachments on those pieces.


\begin{definition}[Permutation graph]
The permutation graph is a directed graph where
the set of nodes correspond to the set of agents. Hence when we refer to the graph, we will use agents and nodes interchangeably.  The arcs of the permutation graph depend on the current state of the working set of isomorphic snapshots $S$. In particular, they depend on which extracted pieces have been attached to the originally allocated pieces in the isomorphic snapshots in $S$. 
For a family in $T$, the edge to its holder comes from the agent whose trim ends the offered prefix, initially its original owner. A family in $T'$ has an edge from every agent, since every agent can take the prefix ending at his own trim.
We build the initial permutation graph with each node pointing only to itself. 
Throughout the protocol, we ensure that each node in the permutation graph has in-degree at least one.
\end{definition}

The permutation graph itself gets updated when isomorphic extractions are attached to a isomorphic pieces in $S$. The process of extracted pieces being attached to an allocated piece results in the aggregated piece becoming attractive to a new agent who then wants to point to the agent holding that piece (see Figures~\ref{fig:perm1} and ~\ref{fig:perm2}). 

The permutation graph also suggests a natural way to exchange pieces. It has a similar idea as the trading graph used in top trading cycles algorithm for housing markets. If there is a cycle in the graph, we have the possibility of exchanging the allocations of agents in the cycle by giving an agent the piece of the agent he points to~\cite{SoUn10a}. The permutation graph is more intricate because updates on the permutation graph reflect simultaneous updates on the working set of isomorphic snapshots $S$. Also by Remark~\ref{remark:consistent}, an agent in an exchange does not offer what he holds but can offer more because of the additional attachments beyond his own extraction. 

\begin{figure}[h!]
    \scalebox{0.5}{
\begin{tikzpicture}

\node[circle,draw] (n1) at (0,0) {1};
\node[circle,draw] (n2) at (3,0) {2};
\node[circle,draw] (n3) at (6,0) {3};
\node[circle,draw] (n4) at (9,0) {4};

\path (n1) edge [thick,bend left=90] (0,2);
\path[->] (0,2) edge [thick,bend left=90] (n1);
\path (n2) edge [thick,bend left=90] (3,2);
\path[->] (3,2) edge [thick,bend left=90] (n2);
\path (n3) edge [thick,bend left=90] (6,2);
\path[->] (6,2) edge [thick,bend left=90] (n3);
\path (n4) edge [thick,bend left=90] (9,2);
\path[->] (9,2) edge [thick,bend left=90] (n4);

\draw (10.5,0)--(10.5,1)--(14,1);
\draw(14,0)--(10.5,0);
\draw(14,-0.5)--(14,-0);
\draw (14,1)--(14,1.5)--(15,0.5)--(14,-0.5);
\node at(12.5,0.5) {Update};
\begin{scope}[shift={(16,0)},rotate=0]
\node[circle,draw] (n1) at (0,0) {1};
\node[circle,draw] (n2) at (3,0) {2};
\node[circle,draw] (n3) at (6,0) {3};
\node[circle,draw] (n4) at (9,0) {4};

\path (n2) edge [thick,bend left=90] (3,2);
\path[->] (3,2) edge [thick,bend left=90] (n2);
\path (n3) edge [thick,bend left=90] (6,2);
\path[->] (6,2) edge [thick,bend left=90] (n3);
\path[->] (n2) edge [thick,bend left=30] (n1);
\path (n4) edge [thick,bend left=90] (9,2);
\path[->] (9,2) edge [thick,bend left=90] (n4);
\end{scope}
\end{tikzpicture}
}
 \caption{Update on the permutation graph after the leftover isomorphic pieces of cake associated with node $1$ and held by agent $1$ have had their first associated piece of cake attached, making the pieces desirable to agent $2$.}
 \label{fig:perm1}
\end{figure}
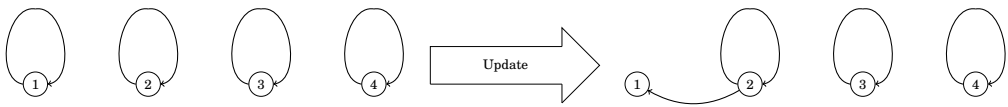

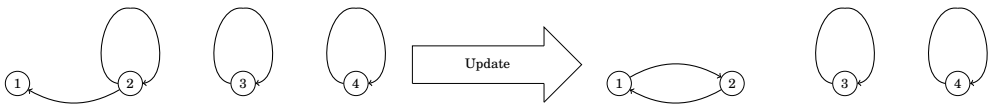
\begin{figure}[h!]
    \scalebox{0.5}{
\begin{tikzpicture}
\node[circle,draw] (n1) at (0,0) {1};
\node[circle,draw] (n2) at (3,0) {2};
\node[circle,draw] (n3) at (6,0) {3};
\node[circle,draw] (n4) at (9,0) {4};

\path (n2) edge [thick,bend left=90] (3,2);
\path[->] (3,2) edge [thick,bend left=90] (n2);
\path (n3) edge [thick,bend left=90] (6,2);
\path[->] (6,2) edge [thick,bend left=90] (n3);
\path (n4) edge [thick,bend left=90] (9,2);
\path[->] (9,2) edge [thick,bend left=90] (n4);
\path[->] (n2) edge [thick,bend left=30] (n1);

\draw (10.5,0)--(10.5,1)--(14,1);
\draw(14,0)--(10.5,0);
\draw(14,-0.5)--(14,-0);
\draw (14,1)--(14,1.5)--(15,0.5)--(14,-0.5);

\node at(12.5,0.5) {Update};
\begin{scope}[shift={(16,0)},rotate=0]
\node[circle,draw] (n1) at (0,0) {1};
\node[circle,draw] (n2) at (3,0) {2};
\node[circle,draw] (n3) at (6,0) {3};
\node[circle,draw] (n4) at (9,0) {4};

\path[->] (n1) edge [thick,bend left=30] (n2);
\path (n3) edge [thick,bend left=90] (6,2);
\path[->] (6,2) edge [thick,bend left=90] (n3);
\path[->] (n2) edge [thick,bend left=30] (n1);
\path (n4) edge [thick,bend left=90] (9,2);
\path[->] (9,2) edge [thick,bend left=90] (n4);
\end{scope}
\end{tikzpicture}
}
 \caption{Update on the permutation graph after the leftover isomorphic pieces of cake associated with node $2$ and held by agent 2  have had their first associated piece of cake attached, making the piece desirable to agent $1$. This results in a non-trivial cycle that allows for agents 1 and 2 to exchange their currently held pieces.}
 \label{fig:perm2}
\end{figure}


\subsection{The GoLeft Protocol}

Now that we have established the important ideas in the GoLeft Protocol, let us tie up the ideas and give an overview of how the protocol works. 
When the GoLeft Protocol starts, it first identifies a working set $S$ of $D$ isomorphic Core snapshots among the $C'$ snapshots. The protocol then constructs a permutation graph corresponding to the working set of isomorphic snapshots. 

In the permutation graph, each node $i$ corresponds to an agent $i$ who holds a set of isomorphic pieces along with their attached extracted pieces in the working set of isomorphic snapshots $S$.
We divide the nodes of the permutation graph into sets $T$ and $T'$. Set $T$ is the set of nodes/agents such that the isomorphic pieces held by them in $S$ have not had $n-1$ attachments).  $T'$ is the set of nodes/agents such that the isomorphic pieces held by them in $S$ have had $n-1$ attachments. 

The protocol identifies a cycle in the permutation graph that includes at least one node $i$ from $T$. Such a cycle always exists. In each of the working set $S$ of isomorphic snapshots, we implement an exchange of pieces held by agents in the cycle: each agent in the cycle is given the piece corresponding to the node that the agent points to in the cycle. After implementing the exchange, the permutation graph is updated to reflect the exchange. If in the exchange, if an agent gets an inferior piece, he always gets the additional extracted pieces associated with the inferior piece up till the agent's extractions. Hence in each snapshot in $S$, each agent's value from his allocation in the snapshot is preserved even if he gets a different piece than in the original Core snapshot.  For any agent $i$, as long as no agent gets extracted pieces beyond $i$'s extraction, $i$ will not be envious. In the GoLeft protocol, it can be the case that some agent $j$ gets extracted pieces beyond $i$'s extracted pieces but before any such attachments in the last part of the GoLeft protocol, we ensure that no envy arises. 

After implementing the cycle, we focus on a node $i\in T$ that was in the cycle.
For agent/node $i$  we know that for all snapshots in the working set $S$, agent $i$ has been allocated the original isomorphic pieces $c_k$ as well as all associated pieces up till $i$'s extracted piece. If the piece of cake agent $i$ is currently allocated in the snapshots $S$ has no more extracted pieces left to attach to it, but it has not had $n-1$ attachments, this means that all agents who have not had their corresponding piece extracted/attached have a significant advantage over agents who have had an extracted piece attached. When this family is selected on the cycle, GoLeft returns its original holder and extractors to Main. After the subsequent Core calls, these agents can divide the residue among themselves.

In case node $i$ does not lead to an exit from the GoLeft Protocol, we know that there are associated pieces that can still be attached to the isomorphic pieces held  by $i$ in the working set of Core snapshots $S$. We focus on the next set of associated pieces $e_{k(l+1)}$ that we are interested to attach to the pieces $c_k$ that have already had associated pieces $e_{k2}, e_{k_3}, \ldots, e_{kl}$ attached in their corresponding main pieces $c_k$. Additionally attaching pieces $e_{k(l+1)}$ to pieces $c_k$ is useful in making the agent who extracted them, interested in the pieces $c_k$ because of the additional $e_{k(l+1)}$ as well as the previous attachments. However, naively attaching the pieces can be problematic and spoil the envy-freeness of the allocation that we maintain. We deal with the issue as follows. 

\begin{itemize}
    \item The agents who did not extract pieces associated with the $c_k$ pieces as well as agents who extracted pieces that have not been attached are asked to `reserve' a big enough subset $S'\subset S$ of snapshots in which they 
value the difference between their bonus value for $c_k$ and the extracted pieces currently attached to $c_k$ the most.
These snapshots $S'$ are removed from $S$ and their remaining unattached associated pieces sent back to the residue. 
By maintaining the advantages in the snapshots $S'$, such agents will not be envious even if some agent in $\{1,\ldots, l\}$ additionally gets all other extracted pieces $e_{k(l+1)}$ in the remaining snapshots in $S$. The advantages reserved due to the discarded snapshots in $S'$ are crucial for the next phase where agents in $\{1,\ldots, l\}$ get extracted pieces $e_{k(l+1)}$ from subset $S''$ of the working set. 
\item The agents indexed from 1 to $l$  who have all already had their extracted pieces attached to $c_k$ are asked to choose a high enough fraction of the snapshots in $S$ in which they value the $e_{k(l+1)}$ pieces. We call these snapshots $S''$.
The $e_{k(l+1)}$ pieces from $S''$ are bunched together and the Main Protocol is called to divide this cake in an envy-free way among the agents indexed from $1$ to $l$ where $l$ is strictly less than $n$.
Since envy-freeness implies proportionality, they derive enough value that they will not be envious if the agent indexed $l+1$ gets all other pieces in set $e_{k(l+1)}$. 
The corresponding set of snapshots $S''$ are then discarded.
\end{itemize}

Hence each time we attach isomorphic extracted pieces $e_{k(l+1)}$ to isomorphic pieces $c_k$, we discard snapshots $S'\cup S''$ from the working set $S$ and still maintain an envy-free allocation. 
Note that in the snapshots that remain in $S$, agents may currently hold a different isomorphic piece than they previously did, but since they also hold the corresponding attachments associated with the isomorphic piece, each agent's total value in each isomorphic snapshot in $S$ stays the same.
In Figure~\ref{fig:dynamics}, we show the states of the permutation graph and the corresponding representative Core snapshot as well as the corresponding extracted pieces.


When the protocol attaches extracted pieces $e_{k(l+1)}$ to allocated pieces $c_k$ currently held by agent $l$, it deletes the incoming edge of node/agent $l$ and replaces it by an edge coming from agent $l+1$ who extracted pieces in  $e_{k(l+1)}$. Intuitively, $l+1$ is now willing to be allocated $c$ and its attached pieces instead of his current pieces in $S$.
We delete previous edges to ensure that until termination, nodes in $T$ have in-degree strictly $1$ which guarantees that no matter the cycle involving a node in $T$ found by the protocol, we will make progress towards termination.  
Eventually the cycle selects a family $c_k$ whose available extractions
have all been offered and for which some agents did not extract. Each
nonextractor began with a significant bonus on this family. The
compensation for earlier exchanges preserves enough of that bonus
against each agent in $A$, the original holder and the extractors.
GoLeft returns $A$ to Main for the Core calls of Lemma~\ref{domin}.

\begin{figure}
\begin{tikzpicture}[scale=0.34]

\node[circle,draw] (n1) at (0,0) {1};
\node[circle,draw] (n2) at (10,0) {2};
\node[circle,draw] (n3) at (20,0) {3};
\node[circle,draw] (n4) at (30,0) {4};
\coordinate (u1) at (0,5){};
\coordinate (u2) at (10,5){};
\coordinate (u3) at (20,5){};
\coordinate (u4) at (30,5){};
\path (n1) edge [red,thick,bend left=90] (u1);
\path[->] (u1) edge [red,thick,bend left=90] (n1);
\path (n2) edge [thick,bend left=90] (u2);
\path[->] (u2) edge [thick,bend left=90] (n2);
\path (n3) edge [thick,bend left=90] (u3);
\path[->] (u3) edge [thick,bend left=90] (n3);
\path (n4) edge [thick,bend left=90] (u4);
\path[->] (u4) edge [thick,bend left=90] (n4);

\draw[fill,blue!60] (0,-2) rectangle (4,-4);
\draw[fill,red!60] (10,-2) rectangle (14,-4);
\draw[fill,green!40] (20,-2) rectangle (24,-4);
\draw[fill,yellow] (30,-2) rectangle (34,-4);
\node at (2,-3){\textbf{1}};
\node at (12,-3){\textbf{2}};
\node at (22,-3){\textbf{3}};
\node at (32,-3){\textbf{4}};

\draw[] (-1,-2) rectangle (-1.5,-4);
\draw[] (-1.5,-2) rectangle (-2,-4);
\draw[] (-2,-2) rectangle (-2.5,-4);
\node at (-1.25,-3) {2};
\node at (-1.75,-3) {3};
\node at (-2.25,-3) {4};

\draw[] (9,-2) rectangle (8.5,-4);
\draw[] (8.5,-2) rectangle (8,-4);
\draw[] (8,-2) rectangle (7.5,-4);
\node at (8.75,-3) {3};
\node at (8.25,-3) {1};
\node at (7.75,-3) {4};

\draw[] (19,-2) rectangle (18.5,-4);
\node at (18.75,-3) {1};

\draw[] (29,-2) rectangle (28.5,-4);
\draw[] (28.5,-2) rectangle (28,-4);
\draw[] (28,-2) rectangle (27.5,-4);
\node at (28.75,-3) {2};
\node at (28.25,-3) {3};
\node at (27.75,-3) {1};

\node at (0,-6){};

\end{tikzpicture}
\begin{tikzpicture}[scale=0.34]

\node[circle,draw] (n1) at (0,0) {1};
\node[circle,draw] (n2) at (10,0) {2};
\node[circle,draw] (n3) at (20,0) {3};
\node[circle,draw] (n4) at (30,0) {4};
\coordinate (u1) at (0,5){};
\coordinate (u2) at (10,5){};
\coordinate (u3) at (20,5){};
\coordinate (u4) at (30,5){};
\path[->] (n2) edge [thick,bend left=30] (n1);
\path (n2) edge [red,thick,bend left=90] (u2);
\path[->] (u2) edge [red,thick,bend left=90] (n2);
\path (n3) edge [thick,bend left=90] (u3);
\path[->] (u3) edge [thick,bend left=90] (n3);
\path (n4) edge [thick,bend left=90] (u4);
\path[->] (u4) edge [thick,bend left=90] (n4);

\draw[fill,blue!60] (0,-2) rectangle (4,-4);
\draw[fill,red!60] (10,-2) rectangle (14,-4);
\draw[fill,green!40] (20,-2) rectangle (24,-4);
\draw[fill,yellow!50] (30,-2) rectangle (34,-4);
\node at (2,-3){\textbf{1}};
\node at (12,-3){\textbf{2}};
\node at (22,-3){\textbf{3}};
\node at (32,-3){\textbf{4}};

\draw[] (-0.5,-2) rectangle (-1,-4);
\draw[] (-1.5,-2) rectangle (-2,-4);
\draw[] (-2,-2) rectangle (-2.5,-4);
\draw (-0.5,-3) -- (0,-3)[densely dotted, thick];
\node at (-0.75,-3) {2};
\node at (-1.75,-3) {3};
\node at (-2.25,-3) {4};

\draw[] (9,-2) rectangle (8.5,-4);
\draw[] (8.5,-2) rectangle (8,-4);
\draw[] (8,-2) rectangle (7.5,-4);
\node at (8.75,-3) {3};
\node at (8.25,-3) {1};
\node at (7.75,-3) {4};

\draw[] (19,-2) rectangle (18.5,-4);
\node at (18.75,-3) {1};

\draw[] (29,-2) rectangle (28.5,-4);
\draw[] (28.5,-2) rectangle (28,-4);
\draw[] (28,-2) rectangle (27.5,-4);
\node at (28.75,-3) {2};
\node at (28.25,-3) {3};
\node at (27.75,-3) {1};

\node at (0,-6){};
\end{tikzpicture}

\begin{tikzpicture}[scale=0.34]

\node[circle,draw] (n1) at (0,0) {1};
\node[circle,draw] (n2) at (10,0) {2};
\node[circle,draw] (n3) at (20,0) {3};
\node[circle,draw] (n4) at (30,0) {4};
\coordinate (u1) at (0,5){};
\coordinate (u2) at (10,5){};
\coordinate (u3) at (20,5){};
\coordinate (u4) at (30,5){};
\path[->] (n2) edge [thick,bend left=30] (n1);

\path[->] (n3) edge [thick,bend left=30] (n2);
\path (n3) edge [red,thick,bend left=90] (u3);
\path[->] (u3) edge [red,thick,bend left=90] (n3);
\path (n4) edge [thick,bend left=90] (u4);
\path[->] (u4) edge [thick,bend left=90] (n4);

\draw[fill,blue!60] (0,-2) rectangle (4,-4);
\draw[fill,red!60] (10,-2) rectangle (14,-4);
\draw[fill,green!40] (20,-2) rectangle (24,-4);
\draw[fill,yellow!50] (30,-2) rectangle (34,-4);
\node at (2,-3){\textbf{1}};
\node at (12,-3){\textbf{2}};
\node at (22,-3){\textbf{3}};
\node at (32,-3){\textbf{4}};

\draw[] (-0.5,-2) rectangle (-1,-4);
\draw[] (-1.5,-2) rectangle (-2,-4);
\draw[] (-2,-2) rectangle (-2.5,-4);
\draw (-0.5,-3) -- (0,-3)[densely dotted, thick];
\node at (-0.75,-3) {2};
\node at (-1.75,-3) {3};
\node at (-2.25,-3) {4};

\draw[] (9,-2) rectangle (9.5,-4);
\draw[] (8.5,-2) rectangle (8,-4);
\draw[] (8,-2) rectangle (7.5,-4);
\draw (9.5,-3) -- (10,-3)[densely dotted, thick];
\node at (9.25,-3) {3};
\node at (8.25,-3) {1};
\node at (7.75,-3) {4};

\draw[] (19,-2) rectangle (18.5,-4);
\node at (18.75,-3) {1};

\draw[] (29,-2) rectangle (28.5,-4);
\draw[] (28.5,-2) rectangle (28,-4);
\draw[] (28,-2) rectangle (27.5,-4);
\node at (28.75,-3) {2};
\node at (28.25,-3) {3};
\node at (27.75,-3) {1};

\node at (0,-6){};
\end{tikzpicture}

\begin{tikzpicture}[scale=0.34]

\node[circle,draw] (n1) at (0,0) {1};
\node[circle,draw] (n2) at (10,0) {2};
\node[circle,draw] (n3) at (20,0) {3};
\node[circle,draw] (n4) at (30,0) {4};
\coordinate (u1) at (0,5){};
\coordinate (u2) at (10,5){};
\coordinate (u3) at (20,5){};
\coordinate (u4) at (30,5){};
\path[->] (n2) edge [red,thick,bend left=30] (n1);
\path[->] (n3) edge [red,thick,bend left=30] (n2);
\path[->] (n1) edge [red,thick,bend left=30] (n3);
\path (n4) edge [thick,bend left=90] (u4);
\path[->] (u4) edge [thick,bend left=90] (n4);

\draw[fill,blue!60] (0,-2) rectangle (4,-4);
\draw[fill,red!60] (10,-2) rectangle (14,-4);
\draw[fill,green!40] (20,-2) rectangle (24,-4);
\draw[fill,yellow!50] (30,-2) rectangle (34,-4);
\node at (2,-3){\textbf{1}};
\node at (12,-3){\textbf{2}};
\node at (22,-3){\textbf{3}};
\node at (32,-3){\textbf{4}};

\draw[] (-1,-2) rectangle (-0.5,-4);
\draw[] (-1.5,-2) rectangle (-2,-4);
\draw[] (-2,-2) rectangle (-2.5,-4);
\draw (-0.5,-3) -- (0,-3)[densely dotted, thick];
\node at (-0.75,-3) {2};
\node at (-1.75,-3) {3};
\node at (-2.25,-3) {4};

\draw[] (9,-2) rectangle (9.5,-4);
\draw[] (8.5,-2) rectangle (8,-4);
\draw[] (8,-2) rectangle (7.5,-4);
\draw (9.5,-3) -- (10,-3)[densely dotted, thick];
\node at (9.25,-3) {3};
\node at (8.25,-3) {1};
\node at (7.75,-3) {4};

\draw[] (19,-2) rectangle (19.5,-4);
\draw (19.5,-3) -- (20,-3)[densely dotted, thick];
\node at (19.25,-3) {1};

\draw[] (29,-2) rectangle (28.5,-4);
\draw[] (28.5,-2) rectangle (28,-4);
\draw[] (28,-2) rectangle (27.5,-4);
\node at (28.75,-3) {2};
\node at (28.25,-3) {3};
\node at (27.75,-3) {1};

\node at (0,-6){};
\end{tikzpicture}

\begin{tikzpicture}[scale=0.34]

\node[circle,draw] (n1) at (0,0) {1};
\node[circle,draw] (n2) at (10,0) {2};
\node[circle,draw] (n3) at (20,0) {3};
\node[circle,draw] (n4) at (30,0) {4};
\coordinate (u1) at (0,5){};
\coordinate (u2) at (10,5){};
\coordinate (u3) at (20,5){};
\coordinate (u4) at (30,5){};
\path (n1) edge [red,thick,bend left=90] (u1);
\path[->] (u1) edge [red,thick,bend left=90] (n1);
\path (n2) edge [thick,bend left=90] (u2);
\path[->] (u2) edge [thick,bend left=90] (n2);
\path (n3) edge [thick,bend left=90] (u3);
\path[->] (u3) edge [thick,bend left=90] (n3);
\path (n4) edge [thick,bend left=90] (u4);
\path[->] (u4) edge [thick,bend left=90] (n4);

\draw[fill,green!40] (0,-2) rectangle (4,-4);
\draw[fill,blue!60] (10,-2) rectangle (14,-4);
\draw[fill,red!60] (20,-2) rectangle (24,-4);
\draw[fill,yellow!50] (30,-2) rectangle (34,-4);
\node at (2,-3){\textbf{1}};
\node at (12,-3){\textbf{2}};
\node at (22,-3){\textbf{3}};
\node at (32,-3){\textbf{4}};

\draw[] (0,-2) rectangle (-0.5,-4);
\node at (-0.25,-3) {1};

\draw[] (10,-2) rectangle (9.5,-4);
\draw[] (8.5,-2) rectangle (8,-4);
\draw[] (8,-2) rectangle (7.5,-4);
\node at (9.75,-3) {2};
\node at (8.25,-3) {3};
\node at (7.75,-3) {4};

\draw[] (20,-2) rectangle (19.5,-4);
\draw[] (18.5,-2) rectangle (18,-4);
\draw[] (18,-2) rectangle (17.5,-4);
\node at (19.75,-3) {3};
\node at (18.25,-3) {1};
\node at (17.75,-3) {4};

\draw[] (29,-2) rectangle (28.5,-4);
\draw[] (28.5,-2) rectangle (28,-4);
\draw[] (28,-2) rectangle (27.5,-4);
\node at (28.75,-3) {2};
\node at (28.25,-3) {3};
\node at (27.75,-3) {1};

\node at (0,-6){};
\end{tikzpicture}

\caption{
    Permutation graph along with corresponding state of an allocation snapshot representative of the working set of isomorphic snapshots. 
 }
\label{fig:dynamics}
\end{figure}
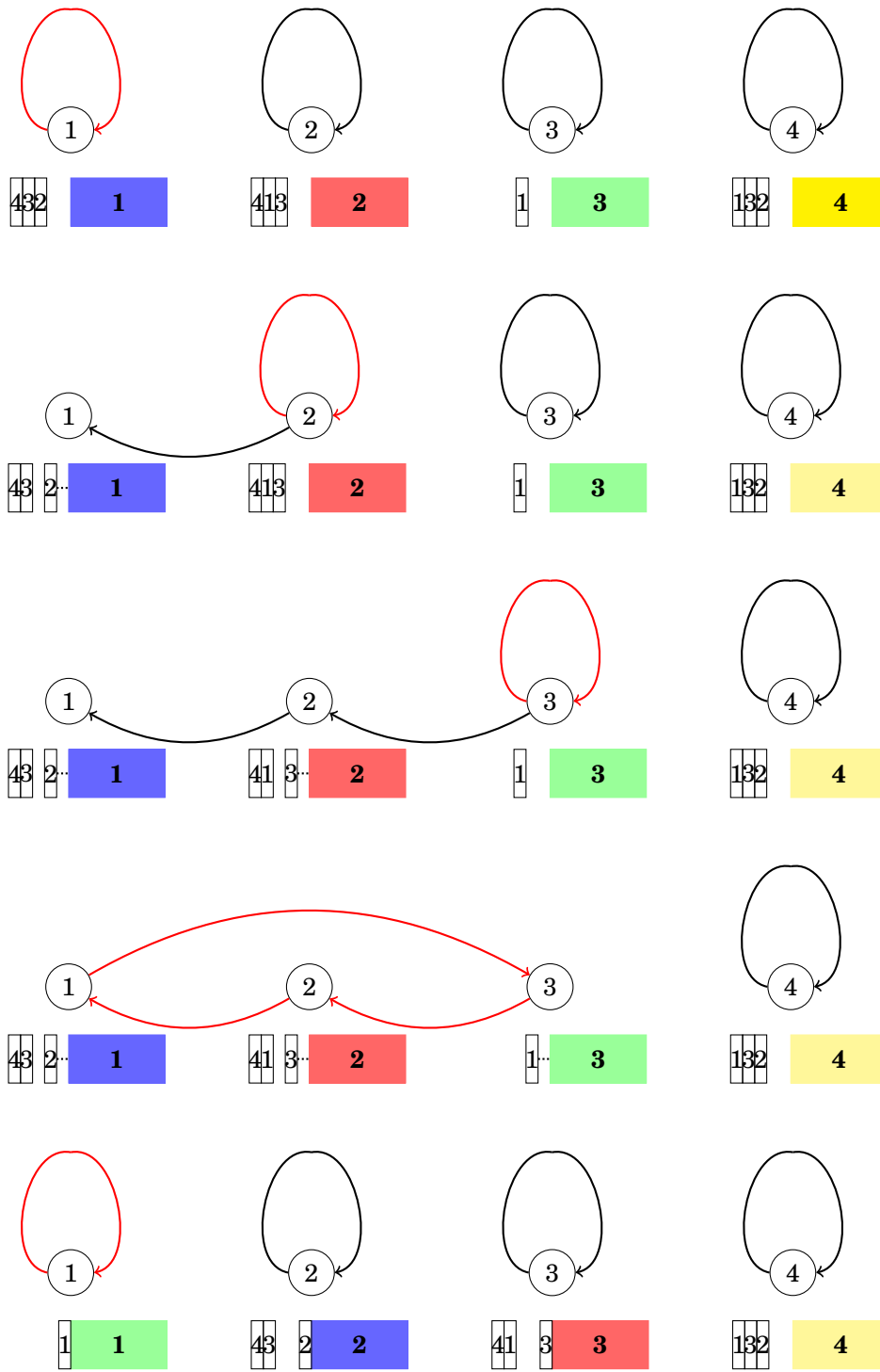

\begin{example}[GoLeft Protocol]
	In Figure~\ref{fig:dynamics}, we demonstrate how the permutation graph along with the  working set of isomorphic snapshots changes in the GoLeft Protocol. Note that even when the representative snapshot changes, there still exist snapshots isomorphic to the previous representative snapshots but these snapshots have been discarded from the working set of snapshots.
    The coloured/shaded pieces represent the pieces given by the Core Protocol to each player. The small pieces on the left of the coloured pieces are extracted pieces, each labelled by the agent who extracted it. At first the extracted pieces are associated with a specific allocated piece. Then they are attached to it (represented by the dotted lines). Finally, when a coloured/shaded piece is reallocated to a new agent, the  extracted pieces attached to it are also allocated to the new agent (in the diagram we now aggregate the extracted piece to the main piece). 
 In the second state of the isomorphic snapshot, agent $2$ points to agent $1$ 
because the piece extracted by agent $2$ has been attached to $1$'s held piece. 
    In the third state of the isomorphic snapshot, agent $3$ points to agent $2$ 
because the piece extracted by agent $3$ has been attached to $2$'s held piece. 
    In the fourth state of the isomorphic snapshot, agent $1$ points to agent $3$ 
because the piece extracted by agent $1$ has been attached to $3$'s held piece.  
In the fifth state, the agents $1,2,3$ exchange their currently held piece and are allocated cake up to their extracted piece. 
   In the fifth (last) state of the isomorphic snapshot, agent $1$ holds a piece up till his extraction but neither agent $2$ or $4$ extracted pieces for the piece that agent $1$ holds.  This means that agents $2$ and $4$ have a significant advantage over agent $1$. Initially the piece was held by $3$ and still is in discarded isomorphic snapshots. This implies significant advantage of $2$ and $4$ over $3$. Therefore agent $2$ and $4$ can be made to dominate $1$ and $3$.

	\end{example}

\section{Argument for Boundedness and Envy-freeness}
In the previous section, we presented the Main Protocol that relies on a few other protocols. We will now argue why it is bounded and envy-free.

\subsection{Argument for Boundedness}
We first show that the Main Protocol is bounded. 

\begin{lemma} \label{Coreb}
  The Core Protocol is bounded by ${(n^2)}^{n+1}$.
\end{lemma}
\begin{proof}
All secondary comparisons use the stored records. Count a marking round
when a group of agents marks the contested labels. A call has at most
$n$ insertion stages. Each stage uses at most one marking round of its
own and at most $n+1$ recursive calls, including the final call.
Every recursive call has fewer participating agents. If $P_d$ bounds
the number of marking rounds over $d$ levels, we may take
\[
 P_0=0,\qquad P_{d+1}=n\bigl(1+(n+1)P_d\bigr).
\]
Induction gives $P_{d+1}+1\leq(n+1)[n(n+1)]^d$. The initial SubCore
call has at most $n-1$ levels, so this recursion counts every marking round.

There are at most $n^2$ trim requests in a round. For each agent we keep
the values of the elementary intervals already cut. The table locates a
new trim within one elementary interval, even if the residue is
disconnected. One Cut and at most $n$ Evaluates suffice to record it.
Thus all the SubCore queries together cost at most
\[
 n^2(n+1)P_{n-1}
 \ \leq\ n^n(n+1)^n
 \ \leq\ 3(n^2)^n,
\]
where $(1+1/n)^n\leq e<3$. The cutter's $n-1$ cuts and their recordings
use at most $(n-1)(n+1)=n^2-1$ further queries. For $n\geq2$,
$3(n^2)^n+n^2\leq(n^2)^{n+1}$. This leaves one query per call to pay
for the $n$ whole-cake evaluations at the start of a sweep through the
cutters. The sweep therefore costs at most $n(n^2)^{n+1}=n^3(n^2)^n$ queries.
\end{proof}
                
\begin{lemma} \label{Discb}
The Discrepancy Protocol is bounded by $Bn\times n^3{(n^2)}^n$.
\end{lemma}
\begin{proof}
Write $f=f(B)$, $e=e_{jkl}$, and let $U$ be the residue at the start of this extraction attempt. The initial sweep starts on $U\setminus e$, since all earlier extractions are restored. At most $M=C'n(n-1)$ chunks had been removed, each worth less than $fV_i(U)$ to every agent $i$. The parameter choices give $Mf<1/(2n)$. Some agent $i$ therefore valued $e$ at least $f(1-Mf)V_i(U)$ when it was considered for extraction. The earlier domination test ensures $V_i(U)>0$. After the sweep $V_i(R)\le fV_i(U)< n V_i(e)$.
For an agent $h$ satisfying the while condition we have $V_h(R)>0$ and $V_h(e)\ge V_h(R)/n$. After $B$ calls with this cutter the new residue is worth at most $fV_h(R)$, where $n^2 f<1$ for $n\ge5$. Thus $V_h(e)>nV_h(R_{\rm new})$ and this agent never triggers the loop again. At exit the initial witness $i$ must satisfy $V_i(e)\ge nV_i(R)$, since his residue value has only decreased. In particular he regards $e$ as significant. If some other agent considers $e$ insignificant, that agent belongs to $D'$, so the split has two nonempty groups. Every member of $D'$ has $V_h(e)<V_h(R)/n$. Otherwise every agent regards $e$ as significant and Main restarts.
There are at most $2nB$ Core calls counting the initial sweep and the loop. Their cost is at most $2nB(n^2)^{n+1}\le Bn\cdot n^3(n^2)^n$ by Lemma~\ref{Coreb}. The comparisons use recorded values.
\end{proof}

\begin{lemma} \label{Gob}
The GoLeft Protocol is bounded by $B'_{n-1}+2Cn^3+2n^2+C'$.
\end{lemma}
\begin{proof}
First let us show that the GoLeft Protocol can execute all steps. The two steps which might be problematic are finding isomorphic snapshots in line \ref{isomorphic} and finding a cycle in line \ref{cycle}. We handle the two cases by a couple of claims.

\smallskip
\begin{claim}\label{claim:find-isomorphic}
Out of the $C'$ snapshots generated, the GoLeft procedure can find $C$ isomorphic snapshots in line \ref{isomorphic}.
\end{claim}
\begin{proof}
For this we use a simple counting argument. The isomorphism is entirely dependent on the order and existence of the extracted pieces. For each allocated piece of cake , there could be up to $n-1$ extracted pieces. Each extracted piece might be extracted by any agent or not exist. This gives us $(n+1)^{n-1}$ possibilities for each allocated piece. Since a snapshot has $n$ allocated pieces, this gives $(n+1)^{n^2-n}$ different possible configurations. We are given $n^{n^{n^{n}}}$ snapshots and desire to find $n^{n^{n}}$ isomorphic snapshots. For $n>4$ we have $C'\geq(n+1)^{n^2-n}C$ (indeed $n+1\leq n^2$ and $n^{n^n}\geq n^n+2n(n-1)$), which implies we can find a set of isomorphic snapshots of size $C$.
\end{proof}
There are enough snapshots to make both selections at every attachment.
Write $d_l=(n-l+1)(ln+1)$. Since $d_l\le n^3$,
$D\le n^{3n(n-1)}\le n^{n^n}=C$ for $n\ge5$.
At index $l$, starting with $s$ snapshots, the first selection leaves
$s/(n-l+1)$ and the second leaves $s/d_l$. Each of the $n$ families
uses index $l$ at most once. Thus the remaining count is always the
product of all unused factors in $D$, including the factors for
attachments which will never be needed. Both selections use integers,
and at least one snapshot remains. 

\smallskip

We now argue that a cycle in line \ref{cycle} can be found.

\begin{claim}\label{claim:cycle-exists}
    In the permutation graph, at each step of the GoLeft Protocol, there exists a cycle containing at least one node from $T$.
\end{claim}


\begin{proof}
    We observe that the permutation graph always contains a cycle since  each node has in-degree at least $1$. For each node $s$ we can travel to the node $t$ pointing to it. This process of backtracking can occur at most as many times as the number of nodes until we cycle. 
    
We now show that it always contains a cycle with at least one node from $T$. Recall that each node in the graph points to nodes in $T'$ if $T'$ is non-empty. Each node in $T$ has in-degree exactly 1 whereas each node in $T'$ has an incoming edge from every agent. We start from any node in $T$ and backtrack according to the incoming arc. In case we only encounter nodes from $T$, we will eventually cycle so that we have found a cycle consisting of nodes only from $T$. Otherwise, we may reach a node $t'$ from $T'$. Since the node $t\in T$ from which we started backtracking points to each node in $T'$ and hence to $t'$, we have again found a cycle consisting of node $t\in T$.
    \end{proof}

%
%


\bigskip 
Now that we handled the two possibly problematic parts of the GoLeft Protocol, let us go over the other steps.
The selection uses the recorded extraction orders and requires no further valuation queries.
The while loop from line \ref{gwhile1} to line \ref{gewhile1} runs until a node in $T$ is found such that it has had less than $n-1$ attachments but has no extracted pieces to be attached (the if condition in step~\ref{step:separation} is true).
To make sure this happens, we require that there always be a node in $T$ so that either we make progress in terms of attaching new sets of extractions or we already encounter the separation condition (Step~\ref{step:separation}). We prove this as follows:

\begin{claim}\label{claim:unfinished}
$T'$ cannot include all nodes.
\end{claim}
\begin{proof}
A family in $T'$ has an extraction from every agent except the original owner. At least one piece of cake must have a significant bonus from the cutter's perspective (see e.g., Remark~\ref{remark:cutteradv}). 
\end{proof}

\bigskip 

 Each continuing iteration offers a new extraction on one family.
There are at most $n(n-1)$ such iterations. All comparisons used in
selecting snapshots, forming the graph and choosing a cycle use recorded
values. Each continuing iteration calls Main at most once on fewer agents.
The induction on the number of agents in the proof of Theorem~\ref{thm:main-bound} gives at most
$B_k^4$ queries on $k\ge5$ agents, where $B_k$ denotes $B$ evaluated
at $k$; the base protocols use fewer than $10^4$ queries. As shown in that proof, at most two physical queries simulate each query of a smaller call on its ordered subcake. Since $2\leq n+1$, the total cost is at most
$n(n-1)(n+1)\max\{10^4,B_{n-1}^4\}\le B'_{n-1}$ for $n\ge5$.
This gives the stated bound. The count for Main uses only the number
of calls per attachment established here.

\end{proof}

\begin{theorem}\label{thm:main-bound}
The Main Protocol is bounded by $B'=n^{n^{n^{n^{n^n}}}}$ Robertson and Webb operations.
\end{theorem}
\begin{proof}
For $n\ge5$ write $f=f(B)$ and $M=C'n(n-1)$, an upper bound on the number of non-owner bonuses in the fixed snapshots. The parameter choices give $Mf<1/(2n)$ and $f<1/n^2$, since $f\le\exp(-2B/n)$ and $B=n^{C'}$. Write $U$ for the residue at the start of an extraction attempt. Each agent values $U$ positively since an agent who values it at zero already dominates all others, causing the earlier test to return. The gap loop puts each bonus either at most $f^2V_i(U)$ or at least $\sqrt f V_i(U)$.

\begin{claim}\label{claim:extract}
Each requested extraction trim exists. In every attempt all accepted extractions together are bounded in value by $MfV_i(U)$ for every agent $i$, where $U$ is that attempt's reference residue.
\end{claim}
\begin{proof}
Each accepted extraction has value less than $fV_i(R)\le fV_i(U)$ to every agent at the time of extraction. There are at most $M$ of them, so throughout the attempt $V_i(R)\ge(1-Mf)V_i(U)$. A small bonus is at most $f^2V_i(U)<fV_i(R)$ even as chunks are removed. A large one is at least $\sqrt f V_i(U)>fV_i(R)$. Thus the same agents qualify throughout for a given snapshot piece. Every requested trim exists, and the chunk ending at the requesting agent's mark, being a subset of his marked prefix, is still insignificant to him when it is considered. This explains why the two cases in the extraction loop exhaust the possibilities.
\end{proof}

After a failed Discrepancy call, let $R_*$ be its residue before the disputed piece $e$ is put back. Its trimmer $u$ has bonus $b\ge V_u(e)$ and $V_u(e)\ge fV_u(R_*)$. Thus $b\ge fV_u(R_*\cup e)/(1+f)$. Since $f/(1+f)>f^2$, this bonus, which was small in the attempt just abandoned, can never again be small at a positive-valued reference residue: all subsequent reference residues lie in $R_*\cup e$. A zero-valued reference residue triggers domination. Thus the number of failed attempts is at most $M$.
Charge each gap-loop Core call in an attempt to its offending bonus.
After $2B$ such calls with the same cutter, the residue has shrunk by
a factor at most $f^2$ in that cutter's eyes. Intervening calls can
only reduce its value further. Thus there are at most $2BM$ gap-loop
Core calls in an attempt. We allow this cost afresh on every restart.

We count the physical queries by induction on the current population $n$.
One agent takes the whole residue. For two and three agents use
cut-and-choose and Selfridge--Conway, respectively.
The base cases with two to four agents use at most $584$ queries on the unit cake.

For $n\geq5$ we prove the stronger bound $B^4$, where $B=n^{C'}$ is
evaluated at $n$. Write $Q=(n^2)^{n+1}$ for the Core bound and let $q$
bound a call with fewer agents on the unit interval. By induction we may
take $q=584$ for $n=5$, and $q=\max\{584,B_{n-1}^4\}$ otherwise.
To run such a call on an ordered subcake, restrict its agents' measures
to that subcake. Cumulative values at recorded component boundaries are
known. A cumulative value inside a component or a cut at a given
cumulative value therefore costs at most one physical query. An interval
Evaluate uses two cumulative values. A Cut uses the cumulative value at
its left endpoint and then a cumulative cut. Thus the simulated call costs
at most $2q$ queries. At its return we intersect each share with the subcake.

Initially we measure the whole cake for each agent and generate $C'$
snapshots. In each of at most $M+1$ attempts allow $2BM+1$ Core calls
for the gap loop, $M$ extraction requests charged $n+1$ queries each,
and $2nB$ Core calls for Discrepancy. The final amplification takes at
most $2nB$ Core calls. There are at most $n(n-1)$ smaller calls within
GoLeft, one per attachment, and at most two additional smaller calls.
The total cost is at most
\[
\begin{aligned}
 & n+C'Q+(M+1)\bigl[(2BM+1)Q+M(n+1)+2nBQ\bigr]\\
 &\hspace{3em}{}+2nBQ+\bigl[2n(n-1)+4\bigr]q .
\end{aligned}
\]
To bound this expression, $C',\,n+1,\,Q,\,q$ are at most $B$.
For $q$ use $4C'_{n-1}\leq C'_n$, which follows from $C'_n=n^{n^{n^n}}$,
and hence $B_{n-1}^4\leq B_n$. Also $(M+1)^2\leq B$ since
$M=C'n(n-1)$ and $B=n^{C'}$. Thus $(M+1)n$, $(M+1)M$, and $M+1$ are
all at most $B$. Expanding the display and applying these inequalities
gives at most $8B^3+3B^2+5B\leq16B^3\leq B^4$. This proves the
induction step. Finally $B^4<n^B=B'$ for $n\geq5$, and the printed
bound covers each base case.
\end{proof}

\subsection{Argument for Envy-freeness}

We now argue for envy-freeness. 

\begin{lemma}\label{lemma:core-neat}
For $n\geq2$, the Core Protocol results in an envy-free partial
allocation in which the cutter cuts the cake into $n$ pieces, each
agent gets a suffix of one of the pieces, and the cutter as well as
at least one non-cutter agent gets a complete piece.
\end{lemma}
\begin{proof}
It suffices to prove that SubCore gives its $n-1$ agents a neat
allocation in which one agent gets a full piece. The cutter can then
take the remaining whole piece. No non-cutter envies it, and the cutter
does not envy any suffix of a piece that he cut equal to his own.

We prove this for every call reached from a zero-benchmark root,
allowing as many input pieces as agents or more. In each call the
allocation must also meet the incoming benchmarks. A whole piece
means the whole \emph{available} input to that call, and containment
of suffixes on one label refers to their common list order. All input
whole entries are formally positive. This holds at the root even for
an empty physical piece. It continues to hold at recursive calls
because, as we show below, only positive-target marks can become
left margins. At entry to each nonroot call we will exhibit distinct
available entries $T_i$ satisfying~\eqref{eq:guard-invariant}.

Induct on the number of agents in the call. Within a call consider
the insertion of agent $m$, assuming the first $m-1$ have been
allocated neatly and meet their benchmarks. If $m$'s favorite piece
at launch is free, give it to him. It is still a favorite by
Claim~\ref{noties}, and it meets his benchmark, which is zero at the
root and is below the value of his available $T_m$ otherwise. Neatness
is preserved and he receives a whole piece. This includes $m=1$ and
the one-agent base case.

Otherwise the $m-1$ allocated labels are contested. Write $X_i$ for
the old allocation of agent $i<m$. Its ownership refinement makes
each $X_i$ strictly preferable for $i$ to every other $X_j$ and every
uncontested whole piece. Write $u_i$ for the value of a favorite
uncontested piece and put $b'_i=\max\{b_i,u_i\}$. There is at least
one uncontested piece, so every target $b'_i$ is positive. Each old
holder values $X_i$ at least $b'_i$ and thus marks its source piece
at or to the right of $X_i$. Each contested piece has a unique
rightmost mark by Claim~\ref{conserv}.

If $b_i>u_i$, the call is nonroot and the input piece containing $T_i$
is contested: its available whole is worth more than $b_i$ to $i$.
Agent $i$ marks strictly to the right of $T_i$, since his target is
$b_i<\widehat V_i(T_i)$. Any other agent $j$ who marks this piece has
target $b'_j\geq b_j>\widehat V_j(T_i)$, so marks strictly to its left.
Thus $i$ is already a winner. In particular the eventual sole outsider
can take a favorite uncontested piece without missing his benchmark.

\begin{claim}\label{existence}
Every recursive call on winners admits the entry refinement of
Section~\ref{tiearg} and returns a neat allocation meeting its benchmarks.
\end{claim}
\begin{proof}
For the first call choose one won piece for each winner and take its
winning mark. These entries have distinct labels, survive restriction
by the non-winner marks, and are worth exactly the winners' targets
to their owners. Any other winner either marked weakly farther left
or valued the whole below his target. Hence the chosen entries are
envy-free. For subsequent calls the entries are the tentative allocation
after the preceding winner admission, proved envy-free in
Claim~\ref{expand} below. The benchmark update copies its owners'
exact values, and ignoring the new winner's marks only enlarges the
available pieces. In either case the ownership refinement now yields
the inequalities~\eqref{eq:guard-invariant} for the recursive call.
It has at most $m-1$ agents, so the outer induction applies. This
includes the final call with $m-1$ agents on $m-1$ pieces.
\end{proof}

\begin{claim}\label{expand}
If $|W|<m-1$, every contested piece left unallocated by the recursive
call has a non-winner mark. Giving its available suffix to the rightmost
non-winner trimmer preserves envy-freeness and the benchmark guarantees.
\end{claim}
\begin{proof}
Let $Y$ be the recursive allocation to $W$. Suppose an unallocated
piece $P$ had no non-winner mark. It would be available in full, and
its old owner $k$ would belong to $W$, since he marked his own label.
Let $H$ consist of the old owners $i\in W$ who value $Y_i$ at least
as much as $X_i$. Neatness gives
$\widehat V_k(Y_k)\geq\widehat V_k(P)\geq\widehat V_k(X_k)$, so $k\in H$.

Map each $i\in H$ to the old owner $j$ of the label of $Y_i$. If
$j\ne i$, strict preference for $X_i$ over $X_j$ gives
$\widehat V_i(Y_i)\geq\widehat V_i(X_i)>\widehat V_i(X_j)$.
Thus $Y_i$ strictly contains $X_j$ on their common label. If $j$ were
a non-winner, his retained mark on that label would restrict $Y_i$
to a suffix of $X_j$. Hence $j\in W$, and envy-freeness of $Y$ gives
$\widehat V_j(Y_j)\geq\widehat V_j(Y_i)>\widehat V_j(X_j)$.
Thus $j\in H$, also trivially true when $j=i$. The map is injective
because $Y$ uses distinct labels. An injection from finite $H$ to
itself is surjective, yet $k$ has no preimage because $P$ is
unallocated. This proves the first assertion.

There is a free label since $|W|<m-1$. Its rightmost non-winner
trimmer $i$ values the available suffix exactly $b'_i$ by
Claim~\ref{noties}. He values every winner's allocation at most
$b'_i$: either his mark restricts the available suffix, or even
the whole piece was below his target. The winners do not envy any
free available piece by neatness. Thus any such admission has the
claimed properties.
\end{proof}

Apply the two claims successively through the loop. The winners'
updated benchmarks never decrease. At the final call every contested
label is allocated, within the sole outsider's mark wherever he
marked. He can therefore take his favorite uncontested piece without
envying a winner. The winners meet their initial $b'_i$, which dominate
every uncontested whole piece; the recursive calls affected none of
those entries. All incoming benchmarks are met, the remaining whole
pieces are unenvied, and the outsider has a whole piece. This completes
both inductions. Projecting the root allocation to physical suffixes
preserves envy-freeness by Claim~\ref{conserv}.
\end{proof}

\begin{example}
We will now walk through the critical part of the SubCore protocol to convey the structure of the proof and make the intuition clearer.
Imagine that we have $5$ agents. A cutter agent has cut the cake into $5$ pieces. The SubCore protocol iteratively introduces each agent with the for loop line~\ref{subcore:forloop}, until all agents have been allocated a piece and the allocation is neat. We will show through this example how introducing a new agent works. Let us say that so far $3$ agents have been introduced and a neat allocation computed for them. We are now adding a $4^{th}$.  The situation can be visualised in Figure~\ref{subcorenewexample1}.
\bigskip

\begin{figure}[h!]
\begin{tikzpicture}[scale=1.1]

\draw (0,0) rectangle (2,1);
\draw (3,0) rectangle (5,1);
\draw (6,0) rectangle (8,1);
\draw (9,0) rectangle (11,1);
\draw (12,0) rectangle (14,1);
\node at (1,-0.5){$a$};
\node at (4,-0.5){$b$};
\node at (7,-0.5){$c$};
\node at (10,-0.5){$d$};
\node at (13,-0.5){$e$};
\node[red] at (1,0.5) {\LARGE $1$};
\node[blue] at (4,0.5) {\LARGE $2$};
\node[green] at (7,0.5) {\LARGE $3$};
\node[black] at (7,2){\LARGE $4$};
\draw [very thick,red] (0.4,1.2)--(0.4,-0.2);

\draw [very thick,blue] (3.3,1.2)--(3.3,-0.2);
\draw [very thick,green] (6,1.2)--(6,-0.2);

\end{tikzpicture}
\caption{}
\label{subcorenewexample1}
\end{figure}

\end{example}

The marked boundaries indicate the tentative allocations in this illustration. In the diagram the trims are of the colour of the agent who placed them. We now query agent $4$ as to which piece he prefers. The easy case is if he prefers piece $d$ or $e$ since then we could give him that piece and neatness would still be satisfied. We therefore dwell on the case where he prefers either $a$, $b$ or $c$. In such a scenario, $a$, $b$ and $c$ will be referred to as the contested pieces. As we will see, one agent will get `kicked out' of these pieces and be forced to pick either $d$ or $e$. Each agent is guaranteed at least his value for his preferred whole piece in $\{d,e\}$. This gives us a lower bound (as long as the assumption made above that only one agent gets kicked out holds) for the value each agent is guaranteed to get. The value an agent is guaranteed $j$ to get is kept track of by the variable $b_j'$ in the protocol. In line \ref{onlytrim}, we ask each agent to place a trim on the contested pieces so that they are equal to this lower bound. This may look as so as in Figure~\ref{fig:subcorenewexample2}.

\begin{figure}
\begin{tikzpicture}[scale=1.1]

\draw (0,0) rectangle (2,1);
\draw (3,0) rectangle (5,1);
\draw (6,0) rectangle (8,1);
\draw (9,0) rectangle (11,1);
\draw (12,0) rectangle (14,1);
\node at (1,-0.5){$a$};
\node at (4,-0.5){$b$};
\node at (7,-0.5){$c$};
\node at (10,-0.5){$d$};
\node at (13,-0.5){$e$};
\node[red] at (5,2) {\LARGE $1$};
\node[blue] at (6,2) {\LARGE $2$};
\node[green] at (7,2) {\LARGE $3$};
\node[black] at (8,2){\LARGE $4$};
\draw [ultra thick,red] (0.6,1.2)--(0.6,-0.2);
\draw [ultra thick,red] (3.6,1.2)--(3.6,-0.2);
\draw [blue] (0.2,1.2)--(0.2,-0.2);
\draw [blue] (3.4,1.2)--(3.4,-0.2);
\draw [green] (6.2,1.2)--(6.2,-0.2);
\draw [green] (3.2,1.2)--(3.2,-0.2);
\draw [ultra thick, black] (6.4,1.2)--(6.4,-0.2);
\node[black] at (6.4,1.4) {\small $b_4'$};
\node[red] at (3.6,1.4)  {\small $b_1'$};
\node[green] at (6.2,-0.4) {\small $b_3'$};
\node[blue] at (3.4,-0.4){\small $b_2'$};
\node[green] at (3.2,1.4) {\small $b_3'$};
\node[red] at (0.6,1.4) {\small $b_1'$};
\node[blue] at (0.2,1.4) {\small $b_2'$};

\end{tikzpicture}
\caption{}
\label{fig:subcorenewexample2}
\end{figure}

In this diagram all the trims have been placed, and the agents who placed the rightmost trims are said to `win' the piece. Their trims are highlighted in the diagram. These agents are added to the set $W$ in line \ref{fwinners}. If each trim had been won by a separate agent, we end up in an easy case of the protocol since each agent except one wins a separate piece, and the one who does not gets kicked out.
However, in the situation depicted by the diagram, agent $1$ has won both pieces $a$ and $b$ and agent $4$ won $c$. We therefore do not know who gets kicked out between $2$ and $3$. However we do know that if we were to continuously make agent $1$ and $4$'s trims move left in an envy-free way, we could eventually ensure that a piece is won by someone new. Of course the protocol cannot move the trims continuously, but it is nonetheless possible to ensure that a new agent gets a contested piece. This is done in line \ref{step:subcore} when we call the protocol recursively. This is what the recursive call would look like in the following situation (Figure~\ref{fig:subcorenewexample3}):

\begin{figure}[h!]
\begin{tikzpicture}[scale=1.1]

\draw (0,0) rectangle (2,1);
\draw (3,0) rectangle (5,1);
\draw (6,0) rectangle (8,1);
\draw[fill,gray] (9,0) rectangle (11,1);
\draw[fill,gray] (12,0) rectangle (14,1);
\node at (1,-0.5){$a$};
\node at (4,-0.5){$b$};
\node at (7,-0.5){$c$};
\node at (10,-0.5){$d$};
\node at (13,-0.5){$e$};
\node[red] at (5,2) {\LARGE $1$};
\node[black] at (8,2){\LARGE $4$};
\draw [ultra thick,red] (0.6,1.4)--(0.6,1);
\node[red] at (0.6, 1.6) {$b_1$};
\draw [ultra thick,red] (3.6,1.4)--(3.6,1);
\node[red] at (3.6, 1.6) {$b_1$};
\draw [blue] (0.2,1.2)--(0.2,-0.2);
\draw[fill,gray] (0,0)rectangle (0.2,1);
\draw[fill,gray] (3,0)rectangle (3.4,1);
\draw [blue] (3.4,1.2)--(3.4,-0.2);

\draw[fill,gray] (6,0)rectangle (6.2,1);
\draw [green] (6.2,1.2)--(6.2,-0.2);
\node[black] at (6.4, 1.6) {$b_4$};
\draw [ultra thick, black] (6.4,1.4)--(6.4,1);

\end{tikzpicture}
\caption{}
\label{fig:subcorenewexample3}
\end{figure}

The greyed out pieces of cake are not called as part of the input. Only the winner agents are called as part of the input (agent $1$ and $4$). We keep track of the value they are guaranteed to get through the inputs $b_1$ and $b_4$. This is important since pieces $d$ and $e$ are no longer part of the input. Now there are two concerns which we may have from this recursive call: 1) Does it actually manage to return a neat allocation, 2) Does the allocation returned allow us to expand the set $W$ of winners. 1) is addressed by Claim \ref{existence}. 2) is addressed by claim \ref{expand}. 
Once the protocol returns, the situation may look like  (Figure~\ref{fig:subcorenewexample4}).
\begin{figure}[h!]
\begin{tikzpicture}[scale=1.1]

\draw (0,0) rectangle (2,1);
\draw (3,0) rectangle (5,1);
\draw (6,0) rectangle (8,1);
\draw(9,0) rectangle (11,1);
\draw (12,0) rectangle (14,1);
\node at (1,-0.5){$a$};
\node at (4,-0.5){$b$};
\node at (7,-0.5){$c$};
\node at (10,-0.5){$d$};
\node at (13,-0.5){$e$};
\node[red] at (5,2) {\LARGE $1$};
\node[blue] at (6,2) {\LARGE $2$};
\node[green] at (7,2) {\LARGE $3$};
\node[black] at (8,2){\LARGE $4$};
\draw [ultra thick,red] (0.2,1.4)--(0.2,1);
\node[red] at (0.2, 1.6) {$b_1'$};
\draw [ultra thick,red] (3.3,1.4)--(3.3,1);
\node[red] at (3.3, 1.6) {$b_1'$};
\draw [blue] (0.2,1.2)--(0.2,-0.2);
\draw [ultra thick,blue] (3.4,1.2)--(3.4,-0.2);

\draw [green] (6.2,1.2)--(6.2,-0.2);
\node[black] at (6.2, 1.6) {$b_4'$};
\draw [ultra thick, black] (6.2,1.4)--(6.2,1);

\end{tikzpicture}
\caption{}
\label{fig:subcorenewexample4}
\end{figure}

The recursive call allocated $a$ to $1$ and $c$ to $4$, leaving $b$ unallocated and with agent $2$'s trim to the right of agent $1$'s updated $b_i'$ value. This means that we can now add agent $2$ to the set of winners.

\begin{lemma}\label{lem:discrepancy-fair}
Assuming envy-freeness is preserved by the Core Protocol, envy-freeness is preserved in the Discrepancy Protocol.
\end{lemma}
\begin{proof}
Every allocation made in Discrepancy is a Core allocation. Adding envy-free allocations on disjoint cake preserves every pairwise envy inequality.
\end{proof}

\begin{lemma}\label{lem:goleft-fair}
Assuming envy-freeness of Core and of Main on fewer than $n$ agents,
GoLeft preserves envy-freeness.
\end{lemma}
\begin{proof}
An agent who changes pieces in an active snapshot receives the new piece
with the strips up to his own mark, as in Remark~\ref{remark:consistent}.
This is worth his original allocation to him, by the definition of his
bonus. His original piece has rank zero and requires no strip. Strips
beyond the current holder's mark remain available for an exchange until
the snapshot leaves the working set.

\citet{Sokolov23} uses ordered pairs of agents to account for
the advantage secured in discarded snapshots in his protocol. We give
the accounting for the exchanges used here. Consider an observer $i$
and another agent $a$. We keep track of the
advantage of $i$ over $a$ already secured in the snapshots removed from
the working set and the cake shared among the earlier extractors.
This is the difference, according to $i$, between their two allocations
there. Removing a snapshot fixes its allocation permanently. For a proof
that covers all subsequent exchanges this advantage must cover the
following sum over \emph{all} active snapshots and piece labels.

For one active snapshot $j$ and label $k$, let $d$ be $i$'s original
bonus on $c_{jk}$. Let $E(t)$ be his value of the first $t$ strips for
this piece, with $E(0)=0$. Write $g$ for the number offered so far.
Write $p_h$ for agent $h$'s rank on this label, zero for its original
owner and $\infty$ for a nonextractor. Thus $d=E(p_i)$ for a finite
rank, and $d\geq E(t)$ at every available rank for a nonextractor,
as verified below. Put $[x]_+=\max\{x,0\}$. Our sum charges the
following amount for this snapshot and label:
\[
  [E(\min\{g,p_a\})-d]_+
  + {\bf1}_{\{p_a<g\}}[d-E(g)]_+.                       \tag{*}
\]
The first term covers any excess value $a$ can hold on this label.
The second keeps some advantage in reserve after $a$ has passed the
label on. Both terms are zero initially.

When a snapshot is removed, its terms disappear from the required sum.
If $a$ actually holds label $k$ there, his prefix ends at his own mark,
no later than $\min\{g,p_a\}$. Since $i$ receives his original value,
the first term for that label covers any loss in fixing this snapshot.
All its second terms are therefore available to spend. Moreover, if
$a$ currently holds the label at rank $g$ and $d\geq E(g)$, fixing
this snapshot earns the additional advantage $d-E(g)$.
A cycle exchange changes none of the terms of~$(*)$.

It remains to check one attachment on the selected label, from $g$ to
$g+1$. The preceding cycle put its rank-$g$ agent on that label,
since its unique incoming edge was from that agent
(the initial case is the owner's self-loop). Its previous holders,
including this one, are exactly the agents of ranks at most $g$.
Write $w=E(g+1)-E(g)$, separately in each snapshot. The first
discarding phase leaves a set $S_1$; the second leaves $S_2\subset S_1$
and shares among these holders the new strips from $S_1\setminus S_2$.
Only the selected label changes its terms on $S_2$.

If $p_i,p_a\leq g$, neither term increases, and both agents participate
in the envy-free sharing. If $p_i,p_a>g$, both terms remain zero and
neither participates. If $p_i\leq g<p_a$, the first term grows by $w$
in each surviving snapshot. In the second phase $i$ selects at least
$n|S_2|$ strips, each at least as valuable to him as any survivor.
His proportional share of the pool among at most $n$ agents is
therefore worth at least $\sum_{j\in S_2}w_j$, and $a$ receives none.

Finally suppose $p_a\leq g<p_i$. Here $d\geq E(g+1)$, and $a$ could
receive the entire pool while $i$ receives none. If $p_a<g$, the second
term released by each second-phase discard is $d-E(g)\geq w$; the
terms on surviving snapshots decrease. If $p_a=g$, then $a$ is the
current holder. In the first phase $i$ selects at least $|S_1|$
snapshots with the largest remaining values of $d-E(g)$. Fixing these
earns at least $\sum_{j\in S_1}(d_j-E_j(g))$ by the preceding
observation. This pays for the entire second-phase pool and for the new
second term $d_j-E_j(g+1)$ on each survivor. This exhausts the cases
and proves the claimed lower bound for the secured advantage.

In particular the sum of first terms covers any loss in all active
snapshots, whoever holds each label. The snapshots outside the initial
working set keep their envy-free allocations. Adding these allocations
gives envy-freeness.
\end{proof}

\begin{lemma}\label{domin}
After GoLeft returns $A$, $2B$ Core calls per cutter suffice to make
every agent outside $A$ dominate every agent in $A$.
\end{lemma}
\begin{proof}
Let $R_0$ be the residue at the start of the successful extraction
pass and $\beta=f(B)$. By the extraction estimate in the proof of boundedness, a nonextractor
$i$ on a label has bonus $d\geq\sqrt\beta\,V_i(R_0)$ throughout this
pass. Each accepted strip is worth at most $\beta V_i(R_0)$ to $i$,
since the residue only decreases within an attempt. There are fewer than
$n$ on one label. Thus for all prefix lengths $t$ on it
\[
 d-E(t)\ \geq\ (\sqrt\beta-n\beta)V_i(R_0)
        \ \geq\ \beta^2V_i(R_0).
\]
For the last inequality, $B\geq 2n^3$ and
$1/\beta=(1+2/(n-2))^B\geq 1+2B/(n-2)\geq4n^2$
for $n\geq5$. In particular all comparisons for nonextractors used
in the preceding proof hold.

At separation, the chosen label has reached its last extraction rank
$g$ and its current holder is the rank-$g$ agent. Its holders, including
the original owner, constitute $A$. Fix $i\notin A$ and $a\in A$.
If $p_a<g$, the second term of~$(*)$ gives at least
$\beta^2 V_i(R_0)$ in each surviving snapshot, besides the first terms
that cover any active envy. If $p_a=g$, $a$ actually holds that label,
so in each surviving snapshot $i$ has this much advantage on the
active allocation itself. At least one snapshot survives.
All cake returned from these extractions lies in $R_0$. The following
$2B$ Core calls with $i$ as cutter reduce his value of the residue to
at most $\beta^2V_i(R_0)$, and every intervening Core allocation
preserves the secured advantage. Thus $i$ dominates $a$.
\end{proof}

\begin{lemma}\label{godom}
Assuming the Main Protocol with $n-1$ or less agents is envy-free, that the Core protocol, the Discrepancy Protocol and the GoLeft Protocol are envy-free, envy-freeness is preserved in the Main Protocol.
\end{lemma}
\begin{proof}
The Main Protocol does not make allocations without calling a sub-protocol. Assuming that all the sub-protocols it calls are envy-free, the fact that it calls them on all agents preserves overall envy-freeness. The Main Protocol also calls itself on a subset of the agents sharing either the residue or some discrepant piece. An early domination test has the same justification as the first case below. The other calls on a proper subset arise as follows. 
\begin{enumerate}
\item  The GoLeft Protocol has terminated, the Core Protocol has been run enough times to ensure that significant bonus values have been converted into domination on the residue $R$, and a subset of agents $N\setminus A$ dominate agents in subset $A$ with respect to the residue $R$. We then run the Main Protocol  with the residue and $A$ as input. Since according to Lemma~\ref{domin} agents in $N \setminus A$ dominate agents in $A$, they are willing to give all the residue to either agents in $A$, therefore no matter how agents in $A$ share $R$, agents in $N \setminus A$ will not be envious.
\item The Discrepancy Protocol has been called and a discrepant piece can be exploited. In that case the Main Protocol will call $2$ instances of itself. The agents in $D$ divide $e_{jkl}$, and the agents in $D'$ divide the returned residue $R$. The loop gives $V_i(e_{jkl})<V_i(R)/n$ for each $i\in D'$. Each agent $i \in D$ thinks that the discrepant piece $e_{jkl}$ is at least $n$ times bigger than $R$. This means that if they think an agent $s$ gets all of $R$ but none of $e_{jkl}$, they will not be envious of him since $i$ is guaranteed $\frac{1}{n}$ value of $e_{jkl}$. Therefore the agents sharing $e_{jkl}$ are not envious of those sharing $R$. The same argument applies the other way around.
\end{enumerate}
\end{proof}

Based on the lemmata above, we have given the overall argument for envy-freeness. In each recursive call of the Main Protocol, the number of agents among which the cake is to be allocated decreases by at least one. Hence when the number of agents is four or less, we can use a known bounded envy-free protocol to allocate the residue in an envy-free manner. 

\begin{theorem}\label{thm:main-fair}
  The Main Protocol allocates all the cake in an envy-free manner.
\end{theorem}

\section{Partial allocations}

Remark~\ref{remark:waste} shows that one Core call per cutter gives an envy-free partial allocation in which each agent receives at least $\nicefrac{1}{n}$ of his value for the whole cake. For $n\geq5$ these are the first $n$ calls of Main.
As a result we answer an open problem posed by Segal-Halevi et al.~\cite{HHA15a} whether there exists a bounded algorithm that returns a proportional and envy-free partial allocation. 
We also note the following.
In both query estimates below, $f(n)$ may be any worst-case query bound for
the SubCore implementation above on $n$ labels and at most $n-1$ participants.
The cost of preparing its input is counted separately. To see why
$n^2=O(f(n))$, let a uniform cutter make $n$ equal intervals.
Give the first $n-2$ noncutters distinct favourite labels, and give the last
noncutter the first agent's favourite. For each noncutter use density $3$
on her favourite, $2$ on the other designated favourites and $1$ on the two
remaining labels. No mark is needed before the last agent arrives.
The first agent then marks every one of the $n-2$ held labels at an interior
point. Each such mark uses one Cut and records the values of all $n$ agents.
Thus $f(n)\geq(n-2)(n+1)$ for $n\geq4$. Including initial measurements,
the first sweep uses at most $n(f(n)+n^2)$ queries.
\begin{lemma}\label{lem:partial-cost}
If the SubCore Protocol uses at most $f(n)$ Robertson--Webb queries, then there exists an algorithm that returns a proportional and envy-free partial allocation using $O(nf(n))$ queries.
\end{lemma}

Note that the partial allocation may give disconnected pieces to the agents. 
We now give an additional result concerning partial allocation of the cake in which agents get connected pieces. 

It is possible to use the SubCore Protocol to obtain  envy-free partial allocations guaranteeing that each agent receives a \emph{connected} piece of value $\nicefrac{1}{3n}$ of the original cake. We specify  Algorithm~\ref{algo:Connected} that achieves this objective.
The initial cake is the interval $[0,1]$. For one agent give it to that agent.
Otherwise, measure each whole-cake value once if it is not already known,
and express each positive valuation as a fraction of that total.
If the total is zero, use interval length as that agent's tie-breaker
throughout ConnectedPieces, including its SubCore calls. The corresponding
cuts and evaluations are computed from coordinates. Such an agent values
every allocation at zero, so the guarantees also hold for his original valuation.
Previously, Segal-Halevi et al.~\cite{HHA15a} showed that in bounded time, one can obtain an envy-free partial allocation in which each agent gets a connected piece of value $1/2^{n-1}$ of the whole cake. Note that our guarantee of $O(\nicefrac{1}{n})$ is optimal up to a constant factor.

\begin{algorithm}[h!]
\caption{ConnectedPieces Protocol}
\label{algo:Connected}
 \begin{algorithmic}
 \scriptsize
 \REQUIRE  Agent set $N$, unallocated cake $R$.
 \ENSURE An envy-free partial allocation of cake for agents in $N$ guaranteeing each agent $\nicefrac{1}{3n}$ of the cake and a connected piece.
 \end{algorithmic}
 \begin{algorithmic}[1]
  \scriptsize

  \FOR{each agent $i\in N$}
      \IF{$i$ is currently not holding a piece}
          \FOR{each connected piece $c_j$ in the residue $R$}
              \STATE Ask $i$ to decompose  $c_j$ into pieces of value $\nicefrac{1}{3n}$ and some leftover piece $\alpha_j$ of value less than $\nicefrac{1}{3n}$.
	   \COMMENT{This does not change the cake except that it has trim marks on it. It is still the same cake with possibly some gaps in it because of agents in the previous iterations holding a piece.}
          \ENDFOR
              \STATE Pick $n$ of the pieces of value $\nicefrac{1}{3n}$ and call the set of those $n$ pieces $R'$
              \STATE Run SubCore($R'$, $N\setminus \{i\}$) and give $i$ one of the unallocated pieces. 
              \STATE Agent $i$ believes the allocation he got from the previous step is of value $\nicefrac{1}{3n}$, if another agent believes that the previous step gave him a bigger piece than what he was previously holding and bigger than $\nicefrac{1}{3n}$ then he keeps that piece, returning any previous piece he was previously holding to the residue. If an agent does not keep the piece given by the SubCore protocol, the piece is returned to the residue $R$.            
      \ENDIF
      
      \COMMENT{At any point of the algorithm, there are agents who do not hold any cake and agents who hold exactly one connected piece. At the start no one holds a piece. Agents who hold a connected piece will keep holding a connected piece. They may exchange their current connected piece with a more preferred connected piece later on. }
  \ENDFOR
  \RETURN envy-free partial allocation (in which each agent gets a connected piece)
 \end{algorithmic}
\end{algorithm}

In each iteration in which the current agent is not already holding a piece, that agent gets one. In particular agent $i$ who is asked to make divisions of value $\nicefrac{1}{3n}$ can get one of the pieces. 
 At any point in the algorithm we should view the cake as having at most $k<n$ gaps or discontinuities as a result of $k$ agents holding connected pieces from the cake. When a new agent $i$ makes trims to create pieces of value $\nicefrac{1}{3n}$, this does not change the cake or the gaps in it.  The algorithm makes at most $n$ calls to the SubCore Protocol.

We observe that for any agent holding a piece of cake, he already gets value at least $\nicefrac{1}{3n}$ of the original cake. For any agent not holding a piece of cake, he values every currently held piece at most $\nicefrac{1}{3n}$. Hence, he thinks that at least $\nicefrac{2}{3}$ value of the cake is unallocated.


\begin{theorem}\label{thm:connected}
The ConnectedPieces Protocol allocates the cake partially in an envy-free way such that each agent has a connected piece of value at least $\nicefrac{1}{3n}$ of the original cake.
\end{theorem}
\begin{proof}
	We show that the algorithm proceeds correctly and terminates with each agent holding one connected piece. The algorithm iteratively calls SubCore with $n$ pieces created by each agent.
The only reason why the algorithm may not proceed is if an agent $i$ cannot produce $n$ divisions of value $\nicefrac{1}{3n}$ of the original cake. When a new agent $i$ is asked to make $n$ divisions of value exactly $\nicefrac{1}{3n}$, we want to show he can make these divisions by appropriate trim marks. Assume that at this point $k$ agents are holding a piece of cake which means that there are $k$ gaps (discontinuities) in the cake. 
So for an agent $i$ who wants to create divisions of value $\nicefrac{1}{3n}$, there can be two problems: (1) too much value of the cake is allocated and (2) the pieces that are allocated lead to discontinuities which prevents the agent from creating $n$ divisions of value $\nicefrac{1}{3n}$. For the first case, agent $i$ values the allocated cake at most $\nicefrac{k}{3n}$. 
Hence at least $\nicefrac{2}{3}$ value of the cake is still unallocated for $i$. This means that there is enough cake to generate $n$ pieces of value $\nicefrac{1}{3n}$.
For the second case, the discontinuities result in at most $\nicefrac{k+1}{3n}$ of the cake not being used. This is because the $k$ gaps divide $R$ into at most $k+1$ pieces and less than $\nicefrac{1}{3n}$ of cake can be lost per piece due to non-divisibility by $\nicefrac{1}{3n}$ of the value of each connected piece. This means that there is still more than $\nicefrac{1}{3}$ of the cake that can be used for the divisions. Since we require $n$ divisions each of value $\nicefrac{1}{3n}$, there is enough cake to obtain these divisions. 

\bigskip

We now argue that the final allocation is envy-free. At some point an agent $i$ gets a piece to hold. He weakly prefers this piece to each piece held immediately before that iteration. In the same iteration $i$ thinks he got at least as valuable a piece as any piece allocated in the iteration because of envy-freeness of SubCore. So $i$ thinks he has a best piece. In case at a latter iteration, he thinks someone got a better piece than his own currently held piece, then due to envy-freeness of SubCore, $i$ can get at least preferred a piece as well. 
\end{proof}

We also note the following.
In each of at most $n$ rounds we run SubCore once. Surveying the gaps
and preparing the pieces, including the initial measurements, brings
the total to at most $n(f(n)+n^2+3n(n+1)+1)$ queries.
\begin{lemma}\label{lem:connected-cost}
If the SubCore Protocol uses at most $f(n)$ Robertson--Webb queries, then there exists an algorithm that uses $O(nf(n))$ queries and allocates the cake partially in an envy-free way such that each agent has a connected piece of value at least $\nicefrac{1}{3n}$ of the original cake. 
\end{lemma}

\section{Discussion}

In this paper we presented the first general discrete and bounded envy-free protocol. 
Now that boundedness has been established, the paper opens the door to new work on finding the optimal bound.  Getting a clearer understanding of the complexity of envy-freeness in terms of the number of queries is an interesting direction for future work.

The Core Protocol we defined can be used as an oracle in other envy-free protocols that allocate a large enough part of the cake. If the Core Protocol is called $n$ times on the updated residue each time with a different agent as cutter, then we obtain an envy-free allocation in which each agent gets at least $\nicefrac{1}{n}$ of his total value for the cake.  
For $n\geq5$, Main therefore identifies a proportional, envy-free partial allocation after its first $n$ Core calls.
The protocol may stop at that point and return this partial allocation.


It will also be interesting to use the new techniques presented in the paper for finding fair allocations for other notions of fairness. 

\paragraph{Formal proofs}
With the conventions clarified above, the accompanying \leavevmode\pdfstartlink attr{/Border [0 0 0]} user{/Subtype /Link /A << /S /URI /URI (https://github.com/SimonWMackenzie/aziz-mackenzie-bounded-lean) >>}Lean~4 development\pdfendlink{} verifies the Main protocol for
every positive number of agents and every admissible answer to a
Robertson--Webb query. Every execution stays within the stated query
bound and returns a complete envy-free allocation. The development also
verifies the proportional partial allocation after the first sweep and the
connected-piece guarantee of Theorem~\ref{thm:connected}. The
\leavevmode\pdfstartlink attr{/Border [0 0 0]} user{/Subtype /Link /A << /S /URI /URI (https://simonwmackenzie.github.io/aziz-mackenzie-bounded-lean/inspector/) >>}interactive proof inspector\pdfendlink{}
links each numbered result to its Lean proof and displays the two side by side. It provides
the complete Lean source.

\section*{Acknowledgments}
The authors thank  Sajid Aziz, Simina Br{\^{a}}nzei,
Omer Lev, Abdallah Saffidine, Erel Segal-Halevi, Xin Huang, Jade Tan-Holmes, Tommaso Urli, and the reviewers of FOCS 2016 for comments. They also thank Steven Brams and Ariel Procaccia for pointers to the literature. 
Haris Aziz is supported by a Julius Career Award.
He thanks Ulle Endriss for introducing the subject to him at the COST-ADT Doctoral School on Computational Social Choice, Estoril, 2010.

\setlength{\bibsep}{3pt}
\bibliographystyle{plainnat}


%

\end{document}